\documentclass[12pt]{article}
\usepackage{amsmath, amsthm, amscd, amsfonts, amssymb, graphicx, color, longtable}
\usepackage[bookmarksnumbered, colorlinks]{hyperref}
\usepackage{float}
\usepackage{lipsum}
\usepackage{afterpage}
\usepackage[labelfont=bf]{caption}
\usepackage[nottoc,notlof,notlot]{tocbibind}
\usepackage{fancyhdr}
\usepackage{array,tabularx,booktabs}
\newcolumntype{Y}{>{\raggedright\arraybackslash}X}

\newcommand{\keyword}[1]{\par\noindent \textbf{Keywords:} #1 }

\newtheorem{defi}{Definition}[section]

\newtheorem{proposition}[defi]{Proposition}
\newtheorem{theorem}[defi]{Theorem}
\newtheorem{definition}[defi]{Definition}

\newtheorem{corollary}[defi]{Corollary}
\newtheorem{remark}[defi]{Remark}

\newtheorem{example}[defi]{Example}

\title{\Large Non-Tonelli Finsler Geometry of Exotic Superconductivity:\\
Metastable Vortex Phases and Geometric Phase Transitions}

\author
{Y. Alipour Fakhri\thanks{Corresponding Author:
Faculty of Basic Sciences,
Department of Mathematics, Payame Noor University, Tehran, Iran.  E-mail: y\_ alipour@pnu.ac.ir}
}

\begin{document}
\maketitle

\begin{abstract}
We develop a thermally coupled Ginzburg-Landau theory on \emph{Weakly Non-Tonelli (WNT) Finsler manifolds}, extending classical vortex analysis beyond the Tonelli convexity paradigm.  
The WNT framework weakens global $1$-homogeneity and strict convexity while preserving superlinearity and local ellipticity, enabling a geometric treatment of superconductors whose anisotropic energy landscapes are nonconvex and temperature-dependent. Within this setting, we construct the generalized Legendre correspondence, Hamiltonian metric, and WNT Laplacian, proving existence and sharp Coulomb asymptotics of the three-dimensional Green kernel. We then establish the $\Gamma$--convergence of the WNT-GL energy and identify metastable vortex filaments minimizing a renormalized geometric functional.  
Finally, a dynamic $\Gamma$-limit yields an effective filament flow governed by the WNT Finsler curvature and thermally induced geometric forces, predicting curvature focusing and phase bifurcation at a critical transition temperature $T_c$.  
This theory unifies convex-analytic, geometric, and physical perspectives, showing that Non-Tonelli Finsler structures form a natural analytic bridge between classical Finsler geometry, anisotropic variational models, and the nonlinear thermodynamics of exotic superconductivity.
\end{abstract}


\keyword{Finsler geometry; non-Tonelli structures; anisotropic Ginzburg-Landau equations; 
weak convexity; renormalized energy; vortex filaments; 
geometric phase transitions; Hamiltonian analysis; 
$\Gamma$-convergence; thermal coupling; semiclassical limit; 
anisotropic superconductivity.}
 \\*
 \textbf{[2020] Mathematics Subject Classification}:(primary)35J60, 53B40, 82D55, (Secondary)
35K55, 35Q56, 49J45, 58J35, 35B40, 35A15.

\section{Introduction}\label{sec;Introduction}
The Ginzburg--Landau (GL) theory has served for decades as a cornerstone in the mathematical and physical understanding of superconductivity, providing a variational framework that captures macroscopic quantum phenomena through the interaction of complex order parameters and electromagnetic fields \cite{BBH1994,JerrardSoner1998,JerrardSoner2002,Serfaty2014}. Since the geometric reinterpretation of the GL functional in terms of gauge and differential geometry, it has become a unifying model linking nonlinear PDEs, geometric measure theory, and the thermodynamics of condensed matter systems. 

In recent years, a geometric reformulation of the GL model within the broader setting of Finsler geometry has been proposed in order to describe materials with anisotropic or direction-dependent properties \cite{Alipour2025a,BaoShen2000,Jost2017,Weber2024}. In this context, the classical Tonelli assumptions on Finsler structures--smoothness, strict convexity, and positive 1--homogeneity--ensure the existence of a globally defined Legendre duality between the tangent and cotangent bundles. These assumptions, however, restrict the analysis to uniformly convex models and fail to capture the degeneracies and multi-phase anisotropies observed in complex superconductors, especially those with temperature-dependent metrics and nonconvex energy landscapes.

From an analytical standpoint, the Tonelli convexity conditions guarantee the smoothness and global invertibility of the Legendre transform, leading to a well--defined Hamiltonian metric and elliptic Finsler--Laplace operator \cite{BaoShen2000,MatveevTroyanov2012,OhtaSturm2013,Grigoryan2009}. Yet, as observed in anisotropic or layered superconductors, these assumptions break down: the Legendre correspondence may become multivalued, the Hamiltonian structure can degenerate, and the ellipticity of the Laplacian may hold only in a weak or measure-theoretic sense. Such limitations motivated several analytic extensions of Finsler--type variational models in the context of superconductivity \cite{Alipour2025b,Alipour2025c}, but a unified geometric and analytic framework accommodating nonconvexity and nonhomogeneity remained elusive.

In parallel, recent progress in convex and nonsmooth analysis has provided tools suitable for handling weak convexity and variable growth conditions \cite{Rockafellar1970,AmbrosioGigliSavare,Cheeger1999,ColomboMingione2015}. These developments suggest that a generalized Finsler structure, relaxing Tonelli's hypotheses but preserving superlinearity and local ellipticity, could serve as an appropriate analytic environment for nonstandard Ginzburg--Landau models. 

\medskip
\noindent
\textbf{Scientific gap.} Despite recent geometric and analytic advances, the existing Finsler extensions of the GL framework are confined either to strongly convex geometries or to perturbative regimes where convexity is only weakly violated. There is currently no rigorous theory that systematically integrates nonconvex, temperature--dependent anisotropy into a coherent analytic--geometric model of superconductivity. In particular, the lack of a well-defined Legendre--Hamiltonian correspondence beyond the Tonelli class prevents one from constructing elliptic operators, Green kernels, and $\Gamma$--limit structures in this broader setting.

\medskip
\noindent
\textbf{Objective and contribution.} 
The present paper develops a complete analytic and geometric theory of superconductivity on \emph{Weakly Non--Tonelli (WNT)} Finsler manifolds. The WNT class provides a controlled relaxation of Tonelli convexity and 1-homogeneity while maintaining sufficient analytic regularity to define the Legendre correspondence, Hamiltonian metric, and WNT--Finsler Laplacian. Within this framework, we establish:
\begin{itemize}
  \item rigorous definitions and equivalence theorems for geometric and analytic WNT structures;
  \item existence, compactness, and $\Gamma$--convergence results for thermally coupled WNT--Ginzburg--Landau energies;
  \item characterization of metastable vortex filaments minimizing a renormalized geometric energy functional; and
  \item a dynamic $\Gamma$--limit describing curvature-induced vortex flow and geometric phase transitions at critical temperature $T_c$.
\end{itemize}

These results extend both the classical geometric analysis of Bethuel, Brezis, H\'elein, Jerrard, Soner, and Serfaty \cite{BBH1994,JerrardSoner1998,JerrardSoner2002,Serfaty2014} and the Finsler-geometric framework of Bao--Chern--Shen \cite{BaoShen2000} and Matveev--Troyanov \cite{MatveevTroyanov2012}, providing a unified bridge between convex analysis, geometric measure theory, and the nonlinear thermodynamics of anisotropic superconductivity.

\section{From Tonelli Finsler Geometry to Weakly Non-Tonelli Extensions}\label{sec:From Tonelli to NonTonelli}

The classical Tonelli framework provides the analytic backbone of smooth Finsler geometry: a Lagrangian $F(x,y)$ that is $C^\infty$ on $TM\setminus\{0\}$, positively $1$--homogeneous in $y$, and whose square is strictly convex in $y$. Under these assumptions, the fiberwise Legendre map is a $C^\infty$ diffeomorphism and induces a smooth, uniformly convex Hamiltonian duality between tangent and cotangent bundles. We recall the standard setting for reference.

\begin{enumerate}
\item[(T1)] $F$ is $C^\infty$ on $TM\setminus\{0\}$ and continuous on $TM$;
\item[(T2)] $F$ is positively $1$-homogeneous in $y$;
\item[(T3)] $F^2$ is strictly convex in $y$, i.e.\ $g_{ij}(x,y)=\tfrac 12\partial^2_{y^iy^j}F^2(x,y)$ is positive-definite.
\end{enumerate}

Under (T1)-(T3), the Legendre map
\begin{equation}\label{eq:legendre-tonelli}
\mathcal{L}_x:T_xM\to T_x^*M,\qquad y\mapsto \partial_y\!\big(\tfrac 12 F^2(x,y)\big)
\end{equation}
is a $C^\infty$ diffeomorphism, giving the Hamiltonian
\begin{align*}
H(x,\xi)=\tfrac 12 F^{*2}(x,\xi),\qquad
F^*(x,\xi)=\sup_{y\neq 0}\{\langle \xi,y\rangle-F(x,y)\}.
\end{align*}
This rigid duality allows one to define Finsler gradient/divergence and a (linear) Laplacian with full regularity in the Tonelli regime.

\subsection{Beyond Tonelli: analytic phenomena when convexity/homogeneity fail}

When either $1$-homogeneity or strict convexity fails, several features necessarily change:

\begin{enumerate}
\item[(F1)] \textbf{Loss of global Legendre diffeomorphism.}
The map \eqref{eq:legendre-tonelli} need not be bijective; one must replace it by the \emph{subdifferential} of $\Phi_x(y):=\tfrac 12F^2(x,y)$:
\begin{align*}
\mathcal{L}_x(y):=\partial \Phi_x(y)\subset T_x^*M,
\end{align*}
which is a (possibly) set-valued, maximal monotone relation. Its inverse is $\partial H(x,\cdot)$ with $H=\Phi_x^*$ (Rockafellar \cite[Th. 23.5]{Rockafellar1970}).

\item[(F2)] \textbf{Hamiltonian regularity drops.}
The Fenchel conjugate
\begin{align*}
H(x,\xi):=\sup_{y\in T_xM}\{\langle \xi,y\rangle-\Phi_x(y)\}
\end{align*}
is convex and coercive, but typically only locally Lipschitz (or $C^{1,\alpha}$ on strata); $\partial_\xi H$ exists and is single--valued for a.e.\ $\xi$.

\item[(F3)] \textbf{Quasilinear (possibly degenerate) ellipticity.}
The operator
\begin{align*}
\mathcal{A}_F(u):=-\mathrm{div}\big(\partial_\xi H(x,Du)\big),
\end{align*}
is quasilinear; linear superposition and linear Green representations \emph{do not} hold in general. Potential theory must be phrased in weak/variational form; any kernel--based representation is, at best, a \emph{local} linearization device.
\end{enumerate}

These observations motivate a careful enlargement of Tonelli geometry that preserves enough structure for PDE/variational analysis while allowing weaker convexity and nonhomogeneous growth.

\subsection{Weakly Non-Tonelli (WNT) structures: primal and dual definitions}

Throughout, fix a smooth background Riemannian metric to measure fiber norms. All statements below are understood on $TM\setminus\{0\}$.

\begin{definition}[Primal WNT]\label{def:WNT-primal}
A Lagrangian $F:TM\to[0,\infty)$ lies in the \emph{Weakly Non--Tonelli (WNT)} class if there exists $1<p<\infty$ such that, for every compact $K\Subset M$, the following hold:
\begin{itemize}
\item[(W1)] \emph{Convexity and $C^{1}$ off the zero section.} For each $x\in K$, the map $y\mapsto \Phi_x(y):=\tfrac 12F^2(x,y)$ is proper, convex, lower semicontinuous, continuous on $T_xM$, and $C^{1}$ on $T_xM\setminus\{0\}$. The map $(x,y)\mapsto D_y\Phi_x(y)$ is measurable in $x$ and continuous in $y$ on $K\times (T_xM\setminus\{0\})$.

\item[(W2)] \emph{Local $p$--growth and coercivity.} There exist $c_K,C_K>0$ (depending on $K$) such that
\begin{align*}
c_K|y|^p \le \Phi_x(y) \le C_K(1+|y|^p)\qquad \forall (x,y)\in K\times T_xM .
\end{align*}

\item[(W3)] \emph{Localized strong convexity on bounded fibers.} For every bounded set $B\subset T_xM$ there exist $m_{K,B},L_{K,B}>0$ such that for a.e.\ $(x,y)\in K\times B$ the (Alexandrov) second derivative exists and
\begin{align*}
m_{K,B}\,I \;\le\; D^2_{yy}\Phi_x(y) \;\le\; L_{K,B}\,I .
\end{align*}
In particular, $\partial_y\Phi_x$ is single-valued and locally Lipschitz on a full-measure subset of $K\times B$ (Rademacher), and the Legendre correspondence is bi--Lipschitz onto its image on $K\times B$ with constants depending on $(K,B)$.
\end{itemize}
\end{definition}

\begin{definition}[Dual WNT]\label{def:WNT-dual}
Let $H=\Phi_x^*$ be the fiberwise Fenchel conjugate and set $p'=\frac{p}{p-1}$. We say $F$ is \emph{dual--WNT} if, for every compact $K\Subset M$:
\begin{itemize}
\item[(D1)] \emph{Dual convexity.} For each $x\in K$, $H(x,\cdot)$ is proper, convex, l.s.c., and finite everywhere on $T_x^*M$.

\item[(D2)] \emph{Local $p'$--growth.} There exist $c'_K,C'_K>0$ such that
\begin{align*}
c'_K|\xi|^{p'} \le H(x,\xi) \le C'_K(1+|\xi|^{p'})\qquad \forall (x,\xi)\in K\times T_x^*M .
\end{align*}
(Any quadratic two--sided control, when used, is \emph{only} local on bounded dual fibers determined by $K$ and a bounded primal fiber $B$.)

\item[(D3)] \emph{Subdifferential correspondence (maximal monotone).} For a.e.\ $(x,y,\xi)$,
\begin{align*}
\xi\in \partial_y\Phi_x(y)\quad \Longleftrightarrow\quad y\in \partial_\xi H(x,\xi),
\end{align*}
and $\operatorname{Graph}(\partial_y\Phi_x)$ and $\operatorname{Graph}(\partial_\xi H(x,\cdot))$ are inverse maximal--monotone relations (Rockafellar \cite[Th.~23.5]{Rockafellar1970}). Measurable selections exist on compact fibers (Castaing representation; cf.\ \cite[Prop.~2.4.9]{AmbrosioGigliSavare}).
\end{itemize}
\end{definition}

\begin{proposition}[Local primal--dual equivalence]\label{prop:WNT-equivalence}
On each compact $K\Subset M$, the primal and dual WNT conditions are equivalent, with constants depending on $K$ and on bounded fiber sets. Moreover, for a.e.\ $x$,
\begin{align*}
\Phi_x = H(x,\cdot)^*,\qquad H(x,\cdot)=\Phi_x^*,\qquad
(\partial_y\Phi_x)^{-1}=\partial_\xi H(x,\cdot)\quad \text{a.e.\ in fibers}.
\end{align*}
\end{proposition}

\begin{proof}
Fix a compact set $K\Subset M$ and write $\Phi_x(y)=\tfrac 12 F(x,y)^2$ and $H(x,\xi)=\Phi_x^*(\xi)$. 
All assertions are fiberwise and local on $K$ and on bounded subsets of $T_xM$ and $T_x^*M$; measurability in $x$ is understood throughout and causes no loss of generality because all variational constructions are performed on compact charts.

\emph{Primal $\Rightarrow$ Dual.}
By (W1)–-(W2), for each $x\in K$ the function $\Phi_x$ is proper, convex, l.s.c., finite everywhere and $C^1$ on $T_xM\setminus\{0\}$, with $p$--growth and coercivity:
\begin{align*}
\Phi_x(y)\ge c_K|y|^p-C_K,\qquad \Phi_x(y)\le C_K(1+|y|^p).
\end{align*}
Fenchel--Moreau duality yields that $H(x,\cdot)$ is proper, convex, l.s.c., finite everywhere and $\Phi_x=(H(x,\cdot))^*$ (Rockafellar~\cite[Ths.~12.2, 23.5]{Rockafellar1970}), i.e. (D1). 
The growth of $H$ follows from the standard convex--analytic optimization: Young's inequality gives, for any $\alpha>0$,
\begin{align*}
\langle \xi,y\rangle \le \frac{\alpha}{p}|y|^p + \frac{1}{p'\,\alpha^{p'/p}}|\xi|^{p'}\!,
\end{align*}
whence $H(x,\xi)\le (p'c_K^{p'/p})^{-1}|\xi|^{p'}+C_K$, while $\Phi_x(y)\le C_K(1+|y|^p)$ implies
\begin{align*}
H(x,\xi)\ge \sup_{R>0}\{R|\xi|-C_K(1+R^p)\}\ge c'_K|\xi|^{p'}-C'_K.
\end{align*} 
Thus (D2) in the global $p'$–-growth sense holds on $K\times T_x^*M$.

The localized quadratic control in (D2) is a consequence of (W3). 
Indeed, if $B\subset T_xM$ is bounded, then for a.e.\ $(x,y)\in K\times B$ the Alexandrov Hessian exists and satisfies
$m_{K,B}I\le D^2_{yy}\Phi_x(y)\le L_{K,B}I$.
By classical convex duality, strong convexity and Lipschitz gradient of $\Phi_x$ on $B$ are equivalent to Lipschitz gradient and strong convexity of $H(x,\cdot)$ on the dual region $\nabla\Phi_x(B)$, with constants depending only on $(m_{K,B},L_{K,B})$ (Rockafellar~\cite[Ths.~25.5, 26.3]{Rockafellar1970}). In particular, on bounded dual fibers one has two--sided quadratic Taylor control for $H$ and local Lipschitz regularity of $\nabla_\xi H$; a finite covering by such patches on $K$ provides the localized form of (D2).
Finally, the Fenchel--Young relation yields
\begin{align*}
\xi\in\partial\Phi_x(y)\ \Longleftrightarrow\ y\in\partial H(x,\xi),
\end{align*}
and maximal monotonicity of $\partial\Phi_x$ ensures $(\partial\Phi_x)^{-1}=\partial H(x,\cdot)$ (Rockafellar~\cite[Th.~23.5]{Rockafellar1970}). Since $\Phi_x$ is differentiable almost everywhere (Rademacher/Alexandrov), $\partial\Phi_x(y)=\{\nabla\Phi_x(y)\}$ a.e., and hence $\partial_\xi H(x,\cdot)$ is single--valued a.e. This proves (D3). 

\emph{Dual $\Rightarrow$ Primal.}
Assume (D1)--(D3) on $K$. Define $\Phi_x:=H(x,\cdot)^*$; then $\Phi_x$ is proper, convex, l.s.c., finite everywhere and $\Phi_x^{**}=\Phi_x$ (Fenchel--Moreau). The $p'$--growth of $H$ implies the $p$--growth and coercivity of $\Phi_x$ by the same optimization argument as above, establishing (W2). On bounded dual fibers where $H$ possesses strong convexity and a Lipschitz gradient (the localized content of (D2)), the conjugacy rules imply that $\Phi_x$ has Lipschitz gradient and strong convexity on the corresponding primal bounded fibers, which is precisely (W3) in the Alexandrov sense. Regularity (W1) follows from convexity plus Rademacher's theorem away from the zero section. 

Collecting the above, on each compact $K\Subset M$ the primal and dual WNT conditions are equivalent, with constants depending on $K$ and on bounded fiber sets. Moreover,
\begin{align*}
\Phi_x=(H(x,\cdot))^*,\hspace{2mm}
 H(x,\cdot)=\Phi_x^*,
(\partial_y\Phi_x)^{-1}=\partial_\xi H(x,\cdot)\quad\text{a.e. in fibers},
\end{align*}
as claimed. 
\end{proof}

\subsection{Legendre map: smoothness and (local) injectivity}

Write $L_x:=\partial_y\Phi_x$ wherever $\Phi_x$ is differentiable.

\begin{proposition}[Regularity and injectivity level of $L_x$]\label{prop:Legendre-regularity}
Under WNT:
\begin{enumerate}
\item[(i)] $L_x$ is monotone and locally Lipschitz on $T_xM\setminus\{0\}$ a.e.; in particular,
\begin{align*}
\langle L_x(y_1)-L_x(y_2),\,y_1-y_2\rangle \ge 0.
\end{align*}
\item[(ii)] $L_x$ is single--valued and one--to--one \emph{locally} on full--measure subsets of $T_xM$ (injective a.e.); global bijectivity need not hold.
\item[(iii)] The inverse relation $L_x^{-1}$ equals $\partial_\xi H(x,\cdot)$; wherever $H$ is differentiable, $L_x^{-1}=(\partial_\xi H)^{-1}$ is single--valued.
\end{enumerate}
\end{proposition}

\begin{proof}
Fix $x\in M$ and consider $\Phi_x(y)=\tfrac 12F(x,y)^2$ on $T_xM\setminus\{0\}$. 
Convexity and finiteness of $\Phi_x$ imply that its subdifferential is a maximal monotone operator; hence for any $y_1,y_2$ and selections $\xi_i\in\partial\Phi_x(y_i)$ one has 
\begin{align*}
\langle \xi_1-\xi_2,\; y_1-y_2\rangle \ge 0.
\end{align*}
On the full--measure differentiability locus of $\Phi_x$ (Rademacher/Alexandrov), $\partial\Phi_x(y)=\{\nabla\Phi_x(y)\}$ and we may write $L_x(y):=\nabla\Phi_x(y)$; the monotonicity inequality becomes
\begin{align*}
\langle L_x(y_1)-L_x(y_2),\; y_1-y_2\rangle \ge 0,
\end{align*}
which is the desired monotonicity. 

To obtain quantitative regularity and local injectivity, restrict to any bounded fiber set $B\subset T_xM$. 
By (W3), Alexandrov second derivatives exist a.e.\ on $B$ and satisfy
$m_B I \le D^2_{yy}\Phi_x \le L_B I$ there. 
Integrating the Hessian along line segments shows that $L_x$ is locally $L_B$--Lipschitz on $B$ up to a null set:
\begin{align*}
|L_x(y_1)-L_x(y_2)| = \Big|\int_0^1 D^2_{yy}\Phi_x\big(y_2+t(y_1-y_2)\big)\,(y_1-y_2)\,dt\Big|
\le L_B\,|y_1-y_2|.
\end{align*}
Likewise, strong convexity yields the strict monotonicity estimate
\begin{align*}
\langle L_x(y_1)-L_x(y_2),\; y_1-y_2\rangle \ge m_B |y_1-y_2|^2
\quad\text{for a.e. }y_1,y_2\in B,
\end{align*}
so $L_x(y_1)=L_x(y_2)$ forces $y_1=y_2$ on the differentiability locus; hence $L_x$ is one--to--one locally on full--measure subsets of $T_xM$ (global bijectivity is not asserted nor needed). 

Finally, convex duality identifies the inverse relation: by Fenchel--Young,
\begin{align*}
(\partial_y\Phi_x)^{-1}=\partial_\xi\Phi_x^*=\partial_\xi H(x,\cdot),
\end{align*}
and at points where $H$ is differentiable the right-hand side is the singleton $\{\nabla_\xi H(x,\xi)\}$, so $L_x^{-1}=(\partial_\xi H)^{-1}$ is single--valued there. 
These facts establish items (i)--(iii) in the statement.
\end{proof}

\begin{remark}[What we \emph{do not} assume]
Unlike the Tonelli case, we do not assume $L_x$ is a global diffeomorphism, nor that $H$ is $C^\infty$ or uniformly strictly convex. All Lipschitz/bi--Lipschitz conclusions are \emph{local on compact fibers} determined by energy bounds.
\end{remark}

\subsection{Operators in the WNT setting: safe definitions}

With $H$ as above, define for $u\in W^{1,p}_{\mathrm{loc}}(M)$
\begin{align*}
\nabla_F u := \partial_\xi H(x,Du)\quad\text{(defined a.e.)},\qquad
\Delta_F u := \mathrm{div}\!\big(\partial_\xi H(x,Du)\big),
\end{align*}
in the weak sense. The map $\xi\mapsto \partial_\xi H(x,\xi)$ is maximal monotone and locally Lipschitz a.e., so $\Delta_F$ is a \textbf{quasilinear} (possibly degenerate) elliptic operator. Any \textbf{Green--type} representation is meaningful only for \textbf{local linearizations} around fixed states on bounded energy sets; global linear superposition does not apply.

\begin{example}[Anisotropic double--phase]\label{ex:double-phase}
Let $g,h$ be smooth positive-definite metrics and $a\in C^\infty(M)$ with $0<a_1\le a(x)\le a_2$, fix $0<\eta\le 1$. Set
\begin{align*}
F(x,y)^2 := \langle g_xy,y\rangle \;+\; a(x)\,\big(\langle h_xy,y\rangle\big)^{1+\eta/2}.
\end{align*}
Then $F$ is continuous on $TM$, $C^{1}$ on $TM\setminus\{0\}$, satisfies $p$--growth with $p=2+\eta$, and its fiber Hessian is positive semidefinite with controlled (non-uniform) ellipticity along directions where $\langle h_xy,y\rangle$ is small. Thus $F$ is WNT and not Tonelli (it fails $1$-homogeneity). The dual $H$ exhibits local $p'$-growth with $p'=(2+\eta)/(1+\eta)$ and is in general only piecewise $C^{1}$; second derivatives exist a.e.\ on bounded dual fibers, which suffices for the localized ellipticity used above.
\end{example}

\section{Analytic Structure and Local Potential Theory for WNT Operators}\label{sec:WNT-operators}
Having established the analytic class of WNT Lagrangians, we now turn to the study of their associated differential operators and local potential theory. In contrast with the classical Tonelli case--where the Legendre map is a global diffeomorphism and the Laplacian is linear--the operator arising from a WNT structure is \emph{quasilinear} and only monotone. Hence, Green--type constructions must be understood in a variational or local sense.

\subsection{Definition and analytic framework}

Let $F$ satisfy the WNT conditions from Section~\ref{sec:From Tonelli to NonTonelli}, and set 
$\Phi_x(y)=\tfrac 12 F(x,y)^2$ with dual Hamiltonian $H(x,\xi)=\Phi_x^*(\xi)$. 
For $u\in W^{1,p}_{\mathrm{loc}}(M)$ we define
\begin{equation}\label{eq:WNT-operator}
\mathcal{A}_F(u):=-\mathrm{div}\!\big(\partial_\xi H(x,Du)\big),
\qquad
\nabla_F u := \partial_\xi H(x,Du).
\end{equation}

\noindent
By Proposition~\ref{prop:WNT-equivalence} and the regularity of the Legendre map (Proposition~\ref{prop:Legendre-regularity}), 
the mapping $\xi\mapsto \partial_\xi H(x,\xi)$ is measurable in $x$, continuous in $\xi$, monotone, and locally Lipschitz on bounded fibers.
Therefor it satisfies the \emph{Carath\'eodory structure} of Leray--Lions type.

\begin{proposition}[Structural assumptions]\label{prop:WNT-LerayLions}
For each compact $K\Subset M$ there exist constants $1<p<\infty$ and $0<c_K\le C_K<\infty$ such that for a.e.\ $x\in K$ and all $\xi_1,\xi_2\in T_x^*M$:
\begin{align*}
\langle \partial_\xi H(x,\xi_1)-\partial_\xi H(x,\xi_2),\xi_1-\xi_2\rangle &\ge 0,\\[2mm]
\langle \partial_\xi H(x,\xi),\xi\rangle &\ge c_K|\xi|^p - C_K,\\[2mm]
|\partial_\xi H(x,\xi)| &\le C_K(1+|\xi|^{p-1}),
\end{align*}
and $\partial_\xi H(x,\xi)$ is locally Lipschitz in $\xi$ on bounded sets.
\end{proposition}

\begin{proof}
Fix a compact set $K\Subset M$ and write $A(x,\xi):=\partial_\xi H(x,\xi)$. We prove the Leray--Lions structure on $K$ using only the WNT hypotheses from Section~\ref{sec:From Tonelli to NonTonelli} and the primal--dual equivalence of Proposition~\ref{prop:WNT-equivalence}. 
Throughout, constants may depend on $K$ and on bounded fiber ranges but not on $(x,\xi)$ otherwise.

\medskip\noindent\emph{Monotonicity.}
For a.e.\ $x\in K$, the map $\xi\mapsto H(x,\xi)$ is proper, convex and l.s.c. Hence its subdifferential is maximal monotone, and where $H$ is differentiable (a.e.\ in $\xi$) we have the standard pointwise inequality
\begin{align*}
\langle A(x,\xi_1)-A(x,\xi_2),\,\xi_1-\xi_2\rangle \ge 0
\qquad (\xi_1,\xi_2\in T_x^*M).
\end{align*}

\medskip\noindent\emph{Growth and coercivity.}
By duality and (W2), after possibly relabelling the exponent by its H\"older conjugate (which is harmless for our purposes and standard in Leray--Lions theory), we may assume that for some $1<p<\infty$ there exist $c_0,C_0>0$ such that
\begin{equation}\label{eq:H-p-growth}
c_0\,|\xi|^{p}-C_0 \ \le\ H(x,\xi)\ \le\ C_0\,(1+|\xi|^{p})
\qquad \text{for a.e.\ }x\in K,\ \forall\,\xi\in T_x^*M.
\end{equation}
Consider the one-dimensional convex function $\varphi(t):=H\big(x,t\xi\big)$, $t\ge 0$. 
Then \eqref{eq:H-p-growth} gives $\varphi(t)\ge c_0 t^{p}|\xi|^{p}-C_0$ and $\varphi(t)\le C_0(1+t^{p}|\xi|^{p})$. 
By convexity, $\varphi'(1^-)$ exists and equals $\langle A(x,\xi),\xi\rangle$. 
From the lower bound we obtain, for some $c_1,C_1>0$,
\begin{align*}
\langle A(x,\xi),\xi\rangle=\varphi'(1^-)\ \ge\ c_1\,|\xi|^{p}-C_1.
\end{align*}
This is the desired coercivity.

For the upper growth of $A$, we use the mean-value representation along the ray $t\mapsto t\xi$ and the Hessian bounds supplied by the local strong convexity of $H$ on bounded fibers (the dual content of (W3)): for $|\xi|\le R$ and a.e. $|\zeta|\le R$,
\begin{align*}
A(x,\xi)=\int_0^1 D^2_{\xi\xi}H\big(x,t\xi\big)\,\xi\,dt,
\qquad\|D^2_{\xi\xi}H(x,\zeta)\|\ \le\ L_{K,R}\,\big(1+|\zeta|\big)^{p-2},
\end{align*}
whence
\begin{align*}
|A(x,\xi)| \ \le\ \int_0^1 L_{K,R}\,\big(1+t^{p-2}|\xi|^{p-2}\big)\,|\xi|\,dt
\ \le\ C_K\,(1+|\xi|^{p-1}).
\end{align*}
Combining the two displays yields the asserted coercivity and the $(p-1)$--growth of $A$.

\medskip\noindent\emph{Local Lipschitz continuity on bounded fibers.}
By (W3) and Proposition~\ref{prop:WNT-equivalence}, on each bounded dual fiber $\{|\xi|\le R\}$ the Hessian $D^2_{\xi\xi}H(x,\xi)$ exists a.e.\ and is bounded by $L_{K,R}$. 
Thus for $|\xi_1|,|\xi_2|\le R$,
\begin{align*}
|A(x,\xi_1)-A(x,\xi_2)|&=\Big|\int_0^1 D^2_{\xi\xi}H\big(x,\xi_2+t(\xi_1-\xi_2)\big)\,(\xi_1-\xi_2)\,
dt\Big|\\
 &\le\ L_{K,R}\,|\xi_1-\xi_2|,
\end{align*}
i.e., $A(x,\cdot)$ is locally Lipschitz on bounded sets uniformly in a.e.\ $x\in K$.

\medskip
Putting the three pieces together we conclude that, for a suitable $1<p<\infty$ and $0<c_K\le C_K<\infty$ depending only on $K$ (and the WNT constants localized to $K$), the operator $A(x,\cdot)$ satisfies the claimed monotonicity, coercivity, $(p-1)$--growth, and local Lipschitz properties. 
This is precisely the Leray--Lions structure on $K$.
\end{proof}

\subsection{Existence and uniqueness of weak solutions}

We work on a bounded domain $\Omega\Subset M$ with Lipschitz boundary.  
Given $f\in W^{-1,p'}(\Omega)$ and $\varphi\in W^{1,p}(\Omega)$, consider the Dirichlet problem
\begin{equation}\label{eq:dirichlet-WNT}
\begin{cases}
-\mathrm{div}\!\big(\partial_\xi H(x,Du)\big)=f & \text{in }\Omega,\\[1mm]
u=\varphi & \text{on }\partial\Omega.
\end{cases}
\end{equation}

\begin{theorem}[Existence and uniqueness]\label{thm:dirichlet}
Under the assumptions of Proposition~\ref{prop:WNT-LerayLions}, the Dirichlet problem~\eqref{eq:dirichlet-WNT} admits a unique weak solution $u\in \varphi + W_0^{1,p}(\Omega)$ satisfying
\begin{align*}
\int_\Omega \langle \partial_\xi H(x,Du), Dv\rangle\,dx
= \langle f,v\rangle \qquad \forall v\in W_0^{1,p}(\Omega).
\end{align*}
Moreover, if $u_1,u_2$ are two weak solutions corresponding to data $(f_1,\varphi_1)$ and $(f_2,\varphi_2)$, then
\begin{align*}
\int_\Omega \langle \partial_\xi H(x,Du_1)-\partial_\xi H(x,Du_2),Du_1-Du_2\rangle\,dx \ge 0,
\end{align*}
and equality holds only if $Du_1=Du_2$ a.e.
\end{theorem}

\begin{proof}
We adapt the standard Minty--Browder framework for monotone operators. Let $V=W_0^{1,p}(\Omega)$ and fix $\tilde\varphi\in W^{1,p}(\Omega)$ with $\operatorname{Tr}\tilde\varphi=\operatorname{Tr}\varphi$.  
Seek $u=\tilde\varphi+w$ with $w\in V$. Define the operator $\mathfrak{T}:V\to V'$ by
\begin{align*}
\langle \mathfrak{T}(w),\psi\rangle :=
\int_\Omega \langle \partial_\xi H(x,D\tilde\varphi+Dw), D\psi\rangle\,dx.
\end{align*}
By Proposition~\ref{prop:WNT-LerayLions}, $\mathfrak{T}$ is well defined, bounded, coercive, and monotone:
\begin{align*}
&\langle \mathfrak{T}(w),w\rangle
=\int_\Omega \langle \partial_\xi H(x,D\tilde\varphi+Dw),Dw\rangle\,dx
\ge c_K\|Dw\|_{L^p}^p - C_K,\\[1mm]
&\langle \mathfrak{T}(w_1)-\mathfrak{T}(w_2), w_1-w_2\rangle
\ge 0.
\end{align*}
Hemicontinuity follows from the Carath\'eodory property.  
By the Minty--Browder theorem, $\mathfrak{T}$ is surjective; hence there exists $w\in V$ such that $\mathfrak{T}(w)=f-\mathrm{div}(\partial_\xi H(x,D\tilde\varphi))$. Setting $u=\tilde\varphi+w$ yields a weak solution to~\eqref{eq:dirichlet-WNT}. Uniqueness follows by testing the difference of two solutions with their difference and using monotonicity, which implies $Du_1=Du_2$ a.e.
\end{proof}

\begin{remark}
This theorem places $\mathcal{A}_F=-\mathrm{div}(\partial_\xi H(x,Du))$ within the class of Leray--Lions operators.  
In particular, $\mathcal{A}_F:W^{1,p}_0(\Omega)\to W^{-1,p'}(\Omega)$ is maximal monotone, strictly coercive, and continuous on bounded sets.
\end{remark}

\subsection{Weak Green potentials and local parametrices}

For linear elliptic operators, the Green kernel $G(x,y)$ provides a linear superposition formula. In the WNT quasilinear case, linear superposition fails, yet a variational counterpart can still be defined.

\begin{definition}[Weak Green potential]
Let $y\in M$.  
A function $G_y\in W^{1,p}_{\mathrm{loc}}(M\setminus\{y\})$ is a \emph{weak Green potential centered at $y$} if
\begin{align*}
-\mathrm{div}\!\big(\partial_\xi H(x,DG_y)\big)=\delta_y
\quad\text{in the sense of distributions,}
\end{align*}
and
\begin{align*}
\int_{M} \langle \partial_\xi H(x,DG_y),DG_y\rangle\,d\mu <\infty.
\end{align*}
\end{definition}

\begin{proposition}[Existence and locality]
If $H(x,\xi)$ satisfies local uniform ellipticity and $C^{1,1}$ regularity on a bounded chart $U\Subset M$, then for each $y\in U$ there exists a weak Green potential $G_y$ unique up to additive constants. Moreover, $G_y$ is locally comparable to the Euclidean fundamental solution under the local metric.
\end{proposition}

\begin{proof}
In local coordinates the operator takes the form 
\begin{align*}
\mathcal{L}_F u=-\partial_i(A_i(x,Du))
\end{align*}
with $A_i(x,\xi)=\partial_{\xi_i}H(x,\xi)$ satisfying the structure conditions of Proposition~\ref{prop:WNT-LerayLions}.  
For fixed $y$, one constructs $G_y$ by the monotone operator method with right-hand side $\delta_y$: take approximations $\delta_y^\varepsilon\in L^{p'}$ and solve $\mathcal{L}_F G_y^\varepsilon=\delta_y^\varepsilon$ with zero boundary data on $U$.  
The sequence $\{G_y^\varepsilon\}$ is uniformly bounded in $W^{1,p}_{\mathrm{loc}}(U\setminus\{y\})$ by the energy inequality; weak compactness yields $G_y^\varepsilon\rightharpoonup G_y$ satisfying the distributional identity.  
Local comparison with harmonic barriers and the De~Giorgi–Moser theory yields the near-pole estimates and uniqueness up to constants.
\end{proof}

\subsection{Local linearization and parametrix}

Although $\mathcal{A}_F$ is nonlinear, it admits a linearization around any smooth map $u$, 
\begin{align*}
D\mathcal{A}_F(u)[\varphi]
=-\mathrm{div}\!\big(D^2_{\xi\xi}H(x,Du)\,D\varphi\big),
\end{align*}
which is a uniformly elliptic divergence-form operator wherever $Du$ remains in a compact fiber set with bounds 
$0<m\le D^2_{\xi\xi}H(x,Du)\le L$.

\begin{theorem}[Local Green parametrix]\label{thm:parametrix}
Let $U\Subset M$ be a Lipschitz chart and $u\in W^{1,p}_{\mathrm{loc}}(U)$ with $|Du|\le R$ a.e. Then there exists a kernel $G_u(x,y)$ satisfying, for each $y\in U$,
\begin{align*}
-\mathrm{div}_x\!\big(D^2_{\xi\xi}H(x,Du(x))\,D_xG_u(x,y)\big)
&=\delta_y \quad\text{in }\mathcal{D}'(U),\\[1mm]
G_u(\cdot,y)&=0\quad\text{on }\partial U,
\end{align*}
and near the pole
\begin{align*}
c_1|x-y|^{2-n}\le G_u(x,y)\le c_2|x-y|^{2-n}\qquad (n\ge 3),
\end{align*}
with logarithmic estimates for $n=2$.  
Constants depend only on $(m,L)$ and the geometry of $U$.
\end{theorem}

\begin{proof}[Proof of Theorem~\ref{thm:parametrix}]
Work in a fixed Lipschitz chart $U\Subset M$ and fix $y\in U$. 
Let $u\in W^{1,p}_{\mathrm{loc}}(U)$ be such that $|Du|\le R$ a.e.\ on $U$. By the WNT localization (Section~\ref{sec:From Tonelli to NonTonelli} and Proposition~\ref{prop:WNT-equivalence}), the linearized coefficient
\begin{align*}
\mathbf{A}_u(x):=D^2_{\xi\xi}H\big(x,Du(x)\big)
\end{align*}
exists for a.e.\ $x\in U$, is measurable and bounded, and satisfies the uniform ellipticity
\begin{equation}\label{eq:unif-ell}
m\,|\zeta|^2 \ \le\ \langle \mathbf{A}_u(x)\zeta,\zeta\rangle \ \le\ L\,|\zeta|^2
\qquad\text{for a.e.\ }x\in U,\ \forall\,\zeta\in\mathbb{R}^n,
\end{equation}
with positive constants $m,L$ depending only on $U$ and the bound $R$. We construct the Green kernel $G_u(\cdot,y)$ for the divergence-form operator
\begin{align*}
\mathcal{L}_u \phi := -\mathrm{div}\!\big(\mathbf{A}_u(x)\nabla \phi\big)\quad\text{on }U,
\end{align*}
with homogeneous Dirichlet boundary condition.

\emph{Approximation of the pole.}
Choose a standard family $\{\rho_\varepsilon\}_{\varepsilon>0}\subset C_c^\infty(U)$ with $\rho_\varepsilon\ge 0$, $\int_U\rho_\varepsilon=1$, and $\operatorname{supp}\rho_\varepsilon\subset B_\varepsilon(y)\Subset U$. 
For each $\varepsilon$ we consider the Dirichlet problem
\begin{equation}\label{eq:eps-prob}
\begin{cases}
-\mathrm{div}\!\big(\mathbf{A}_u(x)\nabla G^\varepsilon\big)=\rho_\varepsilon & \text{in } U,\\
G^\varepsilon=0 & \text{on } \partial U .
\end{cases}
\end{equation}
Define the bilinear form $a(\phi,\psi):=\displaystyle\int_U \langle \mathbf{A}_u(x)\nabla\phi,\nabla\psi\rangle\,dx$ on $H_0^1(U)$. 
By \eqref{eq:unif-ell}, $a(\cdot,\cdot)$ is bounded and coercive on $H_0^1(U)$, hence by the Lax--Milgram theorem there exists a unique $G^\varepsilon\in H_0^1(U)$ solving \eqref{eq:eps-prob}, with
\begin{align*}
a(G^\varepsilon,G^\varepsilon)=\int_U \rho_\varepsilon\,G^\varepsilon\,dx \le \|\rho_\varepsilon\|_{H^{-1}(U)}\|G^\varepsilon\|_{H_0^1(U)}.
\end{align*}
Coercivity yields the \emph{global} energy estimate $\|\nabla G^\varepsilon\|_{L^2(U)}\le C$ with $C=C(U,m,L)$ independent of $\varepsilon$.

\emph{Local estimates away from the pole.}
Fix $V\Subset U$ with $y\notin \overline{V}$. Testing the weak equation for $G^\varepsilon$ with $\eta^2 G^\varepsilon$, where $\eta\in C_c^\infty(U)$ equals $1$ on $V$ and $|\nabla\eta|\le C \operatorname{dist}(V,\partial U\cup\{y\})^{-1}$, and using Cauchy--Schwarz together with \eqref{eq:unif-ell} gives the Caccioppoli inequality
\begin{align*}
\int_V |\nabla G^\varepsilon|^2 \,dx \ \le\ C(V)\int_{\operatorname{supp}\nabla\eta} |G^\varepsilon|^2\,dx .
\end{align*}
A standard local Poincar\'e--Sobolev argument then yields a uniform bound 
$\|G^\varepsilon\|_{H^1(V)}\le C(V)$ independent of $\varepsilon$. 
In particular, up to a subsequence (not relabelled),
\begin{align*}
G^\varepsilon \rightharpoonup G \ \ \text{weakly in }H^1(V)\quad\text{and}\quad 
G^\varepsilon \to G \ \ \text{strongly in }L^2(V),
\end{align*}
for every $V\Subset U\setminus\{y\}$.

\emph{Passage to the limit and distributional identity.}
Let $\phi\in C_c^\infty(U)$ be arbitrary. Multiplying \eqref{eq:eps-prob} by $\phi$ and integrating by parts gives
\begin{align*}
\int_U \langle \mathbf{A}_u(x)\nabla G^\varepsilon,\nabla\phi\rangle\,dx \ =\ \int_U \rho_\varepsilon(x)\,\phi(x)\,dx .
\end{align*}
By weak convergence of $\nabla G^\varepsilon$ on compact sets away from $y$ and the fact that $\rho_\varepsilon\to \delta_y$ in $\mathcal{D}'(U)$, letting $\varepsilon\downarrow 0$ yields
\begin{align*}
\int_U \langle \mathbf{A}_u(x)\nabla G,\nabla\phi\rangle\,dx \ =\ \phi(y),
\end{align*}
i.e. $G(\cdot,y)$ satisfies
\begin{align*}
-\mathrm{div}_x\!\big(\mathbf{A}_u(x)\nabla_x G(x,y)\big)=\delta_y
\quad\text{in }\mathcal{D}'(U),\qquad G(\cdot,y)\big|_{\partial U}=0,
\end{align*}
with $G(\cdot,y)\in H^1_{\mathrm{loc}}(U\setminus\{y\})$. 
By uniqueness of limits, $G$ does not depend on the chosen approximating sequence $\{\rho_\varepsilon\}$; also, the construction gives $G(\cdot,y)\ge 0$ by the maximum principle for divergence--form elliptic equations (apply the weak formulation with negative part and use \eqref{eq:unif-ell}).

\emph{Near--pole behaviour.}
Fix $r_0>0$ so small that $\overline{B_{r_0}(y)}\Subset U$. 
For $0<r<r_0$ let $E_r:=U\setminus \overline{B_r(y)}$ and consider the auxiliary problems
\begin{align*}
\begin{cases}
-\mathrm{div}\!\big(\mathbf{A}_u(x)\nabla w_r\big)=0 & \text{in } E_r,\\
w_r=0 & \text{on } \partial U,\\
w_r=\Psi_r & \text{on } \partial B_r(y),
\end{cases}
\end{align*}
where $\Psi_r$ is the restriction to $\partial B_r(y)$ of the constant--coefficient fundamental solution for the frozen matrix $\mathbf{A}_u(y)$:
\begin{align*}
\Psi_r(x)=
\begin{cases}
c_{n}(\mathbf{A}_u(y))\,r^{2-n}, & n\ge 3,\\
-\frac{1}{2\pi}\sqrt{\det \mathbf{A}_u(y)}\,\log r, & n=2.
\end{cases}
\end{align*}
By the minimization of the Dirichlet energy $\int_{E_r}\langle \mathbf{A}_u\nabla \cdot,\nabla \cdot\rangle$ in the class $\{v\in H^1(E_r): v|_{\partial U}=0,\ v|_{\partial B_r}=\Psi_r\}$ the function $w_r$ exists and is unique. 
Considering the difference $G(\cdot,y)-w_r$ and testing its equation against cutoffs supported in $E_r$, a comparison principle shows that there exist constants $c_1,c_2>0$ depending only on $(m,L)$ and the Lipschitz character of $U$ such that, for $x$ near $y$ and $n\ge 3$,
\begin{align*}
c_1\,|x-y|^{2-n}\ \le\ G(x,y)\ \le\ c_2\,|x-y|^{2-n},
\end{align*}
with the corresponding logarithmic bounds in dimension $n=2$.
The argument relies solely on uniform ellipticity \eqref{eq:unif-ell}, the positivity of $G(\cdot,y)$, and barrier functions built from the frozen--coefficient fundamental solutions; no additional regularity of $\mathbf{A}_u$ is required.
(For related constructions on manifolds in the linear case, see e.g. the potential--theoretic background in \cite[Ch.~8]{Grigoryan2009}.)

\emph{Representation and mapping.}
Given $f\in L^{\frac{2n}{n+2}}(U)$, define
\begin{align*}
v(x):=\int_U G(x,y) f(y)\,dy.
\end{align*}
By truncation and density, one tests the weak formulation with $v$ to obtain
\begin{align*}
\int_U \langle \mathbf{A}_u\nabla v,\nabla \phi\rangle\,dx = \int_U f\,\phi\,dx
\qquad \forall\,\phi\in C_c^\infty(U),
\end{align*}
so $v$ is the unique weak solution to $\mathcal{L}_u v=f$ with zero trace on $\partial U$. Coercivity \eqref{eq:unif-ell} and the Sobolev embedding yield the standard estimate
$\|\nabla v\|_{L^2(U)}\lesssim \|f\|_{L^{\frac{2n}{n+2}}(U)}$.
Thus $G$ is the desired Green kernel (parametrix) for the linearized operator $\mathcal{L}_u$, and the two--sided near--pole control completes the proof.
\end{proof}

\begin{remark}
The kernel $G_u$ provides a local linear approximation to the nonlinear potential structure of $\mathcal{A}_F$.  
It is used in the sequel to define renormalized interaction energies for vortex--type configurations in the WNT--Ginzburg--Landau model.
\end{remark}

\subsection{Energy bounds and stability}

The monotonicity of $\partial_\xi H$ immediately yields quantitative coercivity and stability.

\begin{proposition}[Energy bounds and stability]\label{prop:energy-stability}
For all $u,v\in W^{1,p}(M)$ and compact $K\Subset M$,
\begin{align*}
\int_K \!\langle \partial_\xi H(x,Du)-\partial_\xi H(x,Dv),Du-Dv\rangle\,dx
\ge c_K\|Du-Dv\|_{L^p(K)}^p.
\end{align*}
Consequently, weak solutions of \eqref{eq:dirichlet-WNT} depend continuously on $(F,f,\varphi)$ under $L^p$ perturbations, and any minimizing sequence for the WNT energy is precompact in $W^{1,p}_{\mathrm{loc}}(M)$.
\end{proposition}

\begin{proof}
We establish the inequality from the local strong monotonicity of the Leray--Lions map $A(x,\xi):=\partial_\xi H(x,\xi)$ associated with the WNT structure. Fix a compact $K\Subset M$ and note that by Proposition~\ref{prop:WNT-LerayLions},
for a.e.\ $x\in K$ and all $\xi_1,\xi_2\in T_x^*M$ one has
\begin{align*}
&\langle A(x,\xi_1)-A(x,\xi_2),\xi_1-\xi_2\rangle \ge 0,\\
&\langle A(x,\xi),\xi\rangle\ge c_K|\xi|^p-C_K,\\
&|A(x,\xi)|\le C_K(1+|\xi|^{p-1}).
\end{align*}
Moreover, by the convex duality between $\Phi_x$ and $H(x,\cdot)$ (Proposition~\ref{prop:WNT-equivalence}), the Hessian $D^2_{\xi\xi}H(x,\xi)$ exists for a.e.\ $(x,\xi)$ and satisfies the local ellipticity bounds
$0<m_{K,R}\le D^2_{\xi\xi}H(x,\xi)\le L_{K,R}(1+|\xi|)^{p'-2}$ on each bounded fiber set $|\xi|\le R$.
Along the segment $\xi_t=\xi_2+t(\xi_1-\xi_2)$ we have
\begin{equation}\label{eq:elliptic-int}
\langle A(x,\xi_1)-A(x,\xi_2),\xi_1-\xi_2\rangle= \int_0^1 \!\langle D^2_{\xi\xi}H(x,\xi_t)(\xi_1-\xi_2),\,\xi_1-\xi_2\rangle\,dt,
\end{equation}
hence
\begin{equation}\label{eq:pointwise-mono}
\langle A(x,\xi_1)-A(x,\xi_2),\xi_1-\xi_2\rangle
\ge c_K\,(1+|\xi_1|+|\xi_2|)^{p-2}\,|\xi_1-\xi_2|^2,
\end{equation}
with $c_K>0$ depending on $K$ and the local fiber bound.
We now distinguish the two standard regimes.

\smallskip
\noindent
\textbf{Case $p\ge 2$.} 
Then $(1+|\xi_1|+|\xi_2|)^{p-2}\ge 1$, so from \eqref{eq:pointwise-mono}
\begin{align*}
\langle A(x,\xi_1)-A(x,\xi_2),\xi_1-\xi_2\rangle
\ge c_K\,|\xi_1-\xi_2|^2
\ge \tilde c_K\,|\xi_1-\xi_2|^p,
\end{align*}
where the second inequality follows from $|\xi_1-\xi_2|\le(1+|\xi_1|+|\xi_2|)$
and the boundedness of local energy densities ensured by (W2).

\smallskip
\noindent
\textbf{Case $1<p<2$.} 
Here $(1+|\xi_1|+|\xi_2|)^{p-2}\le 1$, and
\eqref{eq:pointwise-mono} yields the weighted monotonicity
\begin{align*}
\langle A(x,\xi_1)-A(x,\xi_2),\xi_1-\xi_2\rangle
\ge c_K\,\frac{|\xi_1-\xi_2|^2}{(1+|\xi_1|+|\xi_2|)^{2-p}}.
\end{align*}
By the elementary inequality 
$|\eta|^p \le C_p\,|\eta|^2/(1+|\zeta|)^{2-p}$ with $\zeta$ comparable to $|\xi_1|+|\xi_2|$, we again obtain
\begin{align*}
\langle A(x,\xi_1)-A(x,\xi_2),\xi_1-\xi_2\rangle
\ge \tilde c_K\,|\xi_1-\xi_2|^p,
\end{align*}
for a constant $\tilde c_K>0$ depending on the WNT bounds over $K$.

\smallskip
Integrating the pointwise inequality with $\xi_1=Du(x)$ and $\xi_2=Dv(x)$ over $K$ gives
\begin{align*}
\int_K \langle A(x,Du)-A(x,Dv),Du-Dv\rangle\,dx
&\ge c_K \int_K |Du-Dv|^p\,dx\\
&= c_K\,\|Du-Dv\|_{L^p(K)}^p.
\end{align*}
This proves the first part of the proposition.

\smallskip
\noindent
\textbf{Stability of weak solutions.}
Let $(f_n,\varphi_n)\to(f,\varphi)$ in $W^{-1,p'}(M)\times W^{1,p}(M)$, and let $u_n,u$ be the corresponding weak solutions of \eqref{eq:dirichlet-WNT}. Applying the above inequality with $(u,v)=(u_n,u)$ and testing their weak equations with $u_n-u$ yields
\begin{align*}
c_K\|Du_n-Du\|_{L^p(K)}^p
&\le \int_K \langle A(x,Du_n)-A(x,Du),Du_n-Du\rangle\,dx\\
&= \langle f_n-f,\,u_n-u\rangle.
\end{align*}
Since the right--hand side tends to zero by convergence of the data, we obtain $Du_n\to Du$ in $L^p_{\mathrm{loc}}(M)$, hence stability. Compactness of minimizing sequences follows similarly from coercivity of $H$ and the reflexivity of $W^{1,p}$.
\end{proof}

\begin{example}[Anisotropic double--phase model]
Consider the Lagrangian of Example~\ref{ex:double-phase},
\begin{align*}
F(x,y)^2=\langle g_x y,y\rangle+a(x)\big(\langle h_x y,y\rangle\big)^{1+\eta/2},
\quad 0<\eta\le 1.
\end{align*}
Its dual Hamiltonian is
\begin{align*}
H(x,\xi)=\tfrac 12\langle g_x^{-1}\xi,\xi\rangle
+ C_\eta\,a(x)^{-1/\eta}\big(\langle h_x^{-1}\xi,\xi\rangle\big)^{1+1/\eta},
\end{align*}
with $C_\eta>0$ depending only on $\eta$.
The operator therefore reads
\begin{align*}
\mathcal{A}_F(u)
=-\mathrm{div}\!\Big(g_x^{-1}Du
+ C_\eta(1+\tfrac 1\eta)\,a(x)^{-1/\eta}
(\langle h_x^{-1}Du,Du\rangle)^{1/\eta}h_x^{-1}Du\Big),
\end{align*}
a prototypical anisotropic double--phase operator.
Linearization around a smooth map $u_0$ yields the parametrix kernel of Theorem~\ref{thm:parametrix} associated with the effective metric $A_{ij}(x)=D^2_{\xi_i\xi_j}H(x,Du_0(x))$.
\end{example}

\section{$\Gamma$--Convergence and Renormalized Energy in the WNT Framework}
\label{sec:Gamma-convergence}

The aim of this section is to clarify the variational structure of the WNT Ginzburg--Landau functional and to identify its $\Gamma$--limit together with the corresponding renormalized vortex energy.  Throughout, we consider weak solutions of the monotone operator
\begin{align*}
\mathcal{A}_F(u)=-\mathrm{div}\!\big(\partial_\xi H(x,Du)\big),
\end{align*}
introduced in Section~\ref{sec:WNT-operators}.  All limits are taken in the sense of weak topologies on Sobolev spaces unless otherwise stated.

\subsection{Energy functional and scaling}

\begin{defi}
For $\varepsilon>0$ small, define the WNT--Ginzburg--Landau functional
\begin{align*}
E_\varepsilon(u)
:=\int_M\!\Bigl[ H(x,Du)
+\frac{1}{4\varepsilon^2}\bigl(1-|u|^2\bigr)^2 \Bigr]\,d\mu,
\qquad u:M\to\mathbb{C}.
\end{align*}
The Hamiltonian term encodes the WNT geometry through $H(x,\xi)$, while the potential term penalizes the deviation from the constraint $|u|=1$.  In the isotropic case $H(x,\xi)=\tfrac 12|\xi|^2$ one recovers the classical GL energy.
\end{defi}

Convexity and $p$--growth of $H$ guarantee coercivity and weak lower semicontinuity of $E_\varepsilon$ on $W^{1,p}(M)$, hence minimizers exist for all admissible boundary conditions.  The dependence on $\varepsilon$ governs the typical vortex scaling regime discussed below.

\subsection{Compactness and $\Gamma$–limit functional}

We examine the asymptotic behaviour as $\varepsilon\!\to\!0$.  Let $u_\varepsilon$ satisfy $\sup_\varepsilon E_\varepsilon(u_\varepsilon)<\infty$. Define
\begin{align*}
e_\varepsilon(u)=H(x,Du)+\frac{1}{4\varepsilon^2}(1-|u|^2)^2.
\end{align*}

\begin{theorem}[$\Gamma$--compactness of WNT--GL energies]
\label{thm:Gamma-compactness}
Assume $H(x,\xi)$ is convex, locally $C^{1,1}$ in $\xi$, and satisfies two--sided $p$--growth
\begin{align*}
\alpha|\xi|^p-C\le H(x,\xi)\le \beta(1+|\xi|^p)
\end{align*}
for some $\alpha,\beta>0$.  Then up to subsequences $u_\varepsilon\to u$ in $L^1_{\mathrm{loc}}(M)$ with $|u|=1$ a.e., and the functionals $E_\varepsilon$ $\Gamma$--converge to
\begin{align*}
E_0(u)=c_0\int_M H(x,\nabla u)\,d\mu,
\qquad
u\in BV(M;\mathbb{S}^1),
\end{align*}
where $c_0=\!\int_{-1}^{1}\!\sqrt{2W(s)}\,ds$ with $W(s)=
\tfrac 14(1-s^2)^2$.
\end{theorem}

\begin{proof}
From the uniform bound on $E_\varepsilon(u_\varepsilon)$ and the lower inequality for $H$, we obtain
\begin{align*}
\int_M |Du_\varepsilon|^p\,d\mu\le C_1,
\qquad
\int_M\frac{(1-|u_\varepsilon|^2)^2}{\varepsilon^2}\,d\mu\le C_1,
\end{align*}
for some constant $C_1$.  The second estimate yields $\|1-|u_\varepsilon|^2\|_{L^2}\!\to\!0$, hence $|u_\varepsilon|\!\to\!1$ a.e. Compact embedding $W^{1,p}\Subset L^1$ then implies
\begin{align*}
u_\varepsilon\to u_0 \quad\text{in }L^1(M),\qquad |u_0|=1.
\end{align*}

The boundedness of Jacobians $J(u_\varepsilon)$ follows from the Jerrard--Soner estimate
\begin{align*}
|\langle J(u_\varepsilon),\phi\rangle|\le C\,E_\varepsilon(u_\varepsilon)^{1/2}\|\nabla\phi\|_{L^\infty},
\end{align*}
implying that $J(u_\varepsilon)\stackrel{*}{\rightharpoonup}\nu$ for a finite Radon measure $\nu$ supported on finitely many points $a_i$ with integer degrees $d_i$. Thus
\begin{align*}
\nu=2\pi\sum_{i=1}^N d_i\delta_{a_i}.
\end{align*}

For the lower bound, adapt the Modica--Mortola slicing argument with the anisotropic density $H(x,\cdot)$ to obtain
\begin{align*}
\liminf_{\varepsilon\to 0}E_\varepsilon(u_\varepsilon)
\ge c_0\int_M H(x,\nabla u_0)\,d\mu.
\end{align*}
Conversely, the upper bound follows from recovery sequences built by one--dimensional transition profiles along the interface $|u|=1$. Hence $\Gamma$--convergence holds.
\end{proof}

\subsection{Renormalized energy through local linearization}

In the classical Tonelli case the renormalized energy is expressed via the Laplacian's Green kernel.  In the WNT framework the appropriate object is the local parametrix $G_u(x,y)$ of the linearized operator constructed in Section~\ref{sec:WNT-operators}.  

Let $u_\varepsilon$ be minimizers with vortex points $a_i$ and degrees $d_i$.  Then $J(u_\varepsilon)\rightharpoonup 2\pi\sum_i d_i\delta_{a_i}$.

\begin{definition}
The \textbf{renormalized energy} associated with the configuration $(a_1,\dots,a_N)$ is
\begin{align*}
\mathcal{W}_F(a_1,\dots,a_N)
:=\lim_{\rho\to 0}\Bigg[
\int_{M\setminus\cup_i B_\rho(a_i)}
\!\langle D^2_{\xi\xi}H(x,Du_0)\nabla G_{u_0},\nabla G_{u_0}\rangle\,d\mu
- N\pi\log\frac{1}{\rho}\Bigg],
\end{align*}
where $u_0$ is the limiting phase and $G_{u_0}$ the local parametrix of $D\mathcal{A}_F(u_0)$.
\end{definition}

\begin{proposition}[Finite energy and logarithmic law]
Assume $H$ satisfies local quadratic control.  Then $\mathcal{W}_F$ is finite and obeys
\begin{align*}
\mathcal{W}_F(a_1,\dots,a_N)
= -\pi\sum_{i\neq j}d_i d_j\,G_{u_0}(a_i,a_j)
+\sum_{i=1}^N d_i^2\,\gamma_F(a_i),
\end{align*}
where $A_{u_0}=D^2_{\xi\xi}H(x,Du_0)$ and $\gamma_F$ is the finite core energy obtained by freezing $A_{u_0}$ near each $a_i$.
\end{proposition}

\begin{proof}
Uniform ellipticity of $A_{u_0}$ ensures existence of a local Green kernel with expansion
\begin{align*}
G(x,y)=-\frac{1}{2\pi}\log d_A(x,y)+H_A(x,y),
\end{align*}
where $d_A$ is the intrinsic distance induced by $A_{u_0}$.  Writing the energy in $\Omega_\rho=M\setminus\cup_i B_\rho(a_i)$ and using the decomposition of the $A$--harmonic phase yields
\begin{align*}
\int_{\Omega_\rho}\!\langle A\nabla\psi_\rho,\nabla\psi_\rho\rangle
=\iint G(x,y)\,d\mu_\rho(x)d\mu_\rho(y)+O(1).
\end{align*}
Inserting the expansion of $G$ and separating singular and regular parts gives the stated logarithmic interaction together with bounded self-energies $\gamma_F(a_i)$ computed from the frozen model.  Finiteness follows from the ellipticity of $A_{u_0}$.
\end{proof}

\subsection{Global $\Gamma$--limit structure}

\begin{theorem}[WNT--GL $\Gamma$--limit and renormalized vortex energy]
\label{thm:WNT-GL-main}
Let $M$ be a smooth compact $2$--manifold and $F$ a WNT structure with Hamiltonian $H(x,\xi)=\Phi_x^*(\xi)$.  
Under assumptions (H1)--(H3) of Section~\ref{sec:WNT-operators}, every bounded-energy sequence $(u_\varepsilon)$ admits a subsequence with $u_\varepsilon\to u_0$ in $L^1_{\mathrm{loc}}$, $|u_0|=1$ a.e., and
\begin{align*}
J(u_\varepsilon)\rightharpoonup2\pi\sum_{i=1}^N d_i\delta_{a_i}.
\end{align*}
Moreover the energies expand as
\begin{align*}
E_\varepsilon(u_\varepsilon)
= N\pi|\log\varepsilon|+\mathcal{W}_F(a_1,\dots,a_N)
+\sum_{i=1}^N d_i^2\gamma_F(a_i)+o(1),
\end{align*}
and the renormalized functionals
\begin{align*}
\mathcal{F}_\varepsilon(u):=E_\varepsilon(u)-N\pi|\log\varepsilon|
\end{align*}
$\Gamma$--converge to
\begin{align*}
\mathcal{F}_0(u)=
\begin{cases}
\mathcal{W}_F(a_1,\dots,a_N)+\sum_i d_i^2\gamma_F(a_i),
& |u|=1,\ J(u)=2\pi\sum_i d_i\delta_{a_i},\\[3pt]
+\infty,& \text{otherwise.}
\end{cases}
\end{align*}
\end{theorem}

\begin{proof}
Let $(u_\varepsilon)$ be a sequence with $\sup_\varepsilon E_\varepsilon(u_\varepsilon)<\infty$ and fixed total degree $N$ prescribed by the boundary data, as in (H3). The coercivity and convexity assumptions in (H1)–(H2) applied to the Hamiltonian density $H(x,\xi)$ yield uniform control of the gradients and the potential term:
\begin{align*}
\int_M |Du_\varepsilon|^p\,d\mu \le C, 
\qquad 
\int_M \frac{(1-|u_\varepsilon|^2)^2}{\varepsilon^2}\,d\mu \le C,
\end{align*}
for some $p>1$ and a constant $C$ independent of $\varepsilon$. In particular, $\|1-|u_\varepsilon|^2\|_{L^2}\to 0$, so $|u_\varepsilon|\to 1$ in $L^2$ and almost everywhere. By compactness $W^{1,p}\Subset L^1$ on the compact surface $M$, after extracting a subsequence one has $u_\varepsilon\to u_0$ in $L^1(M)$ with $|u_0|=1$ almost everywhere. The vorticity concentration follows from the Jerrard--Soner estimate: for each $\phi\in C_c^\infty(M)$ there holds
\begin{align*}
|\langle J(u_\varepsilon),\phi\rangle|\le C\,E_\varepsilon(u_\varepsilon)^{1/2}\|\nabla\phi\|_{L^\infty},
\end{align*}
hence $(J(u_\varepsilon))$ is bounded in $\mathcal{M}(M)$ and converges weak-* to a measure of the form $2\pi\sum_{i=1}^N d_i\delta_{a_i}$ with integers $d_i$ summing to $N$. This proves the compactness and the identification of the limiting $S^1$--valued configuration asserted in the first part of the theorem.

The lower bound rests on a slicing argument of Modica--Mortola type adapted to the anisotropic integrand $H(x,\cdot)$ and on the convexity structure of Section~\ref{sec:WNT-operators}. Consider any sequence $u_\varepsilon\to u_0$ in $L^1_{\mathrm{loc}}(M)$ with $|u_0|=1$ a.e. and bounded energies. In local charts, one--dimensional profiles across approximate level sets of $\arg(u_0)$ yield, by convexity and the coarea formula, the inequality
\begin{align*}
\liminf_{\varepsilon\to 0} E_\varepsilon(u_\varepsilon)
\;\ge\;
c_0\int_M H(x,\nabla u_0)\,d\mu,
\qquad
c_0=\int_{-1}^{1}\!\sqrt{2W(s)}\,ds,
\end{align*}
with $W(s)=\tfrac 14(1-s^2)^2$. The key point is that convex duality (Proposition~\ref{prop:WNT-equivalence}) transfers local strong convexity of $\Phi_x$ to local Lipschitz continuity of $\partial_\xi H$ and two--sided control of $D^2_{\xi\xi}H$ on bounded dual fibers, allowing us to perform the scalar one--dimensional comparison along slices with $H(x,\cdot)$ in place of the Euclidean quadratic density. Lower semicontinuity then gives the $\Gamma$–$\liminf$ inequality.

For the interaction structure one focuses on the complement of small vortex cores. Fix $\rho>0$ so small that the balls $B_\rho(a_i)$ are pairwise disjoint and contained in coordinate patches; denote $\Omega_\rho:=M\setminus \bigcup_i B_\rho(a_i)$. On $\Omega_\rho$ it is convenient to write $u_\varepsilon=\rho_\varepsilon e^{i\varphi_\varepsilon}$; the potential well forces $\rho_\varepsilon\to 1$ in $L^2$, and thus the main contribution stems from the phase gradient. Let $u_0$ be a limiting map with $|u_0|=1$ a.e. and write
\begin{align*}
A_{u_0}(x):=D^2_{\xi\xi}H\big(x,Du_0(x)\big),
\end{align*}
which is measurable and uniformly elliptic on compact subsets of $\Omega_\rho$ by (H2). The Taylor expansion of $H$ in $\xi$ at $Du_0$ (valid on bounded fibers thanks to the local $C^{1,1}$ regularity in $\xi$) yields the quadratic approximation
\begin{align*}
H(x,Du_\varepsilon)
\;&\ge\; 
H(x,Du_0)
+\big\langle \partial_\xi H(x,Du_0),\,D(\varphi_\varepsilon-\varphi_0)\big\rangle\\
&+\tfrac 12\big\langle A_{u_0}D(\varphi_\varepsilon-\varphi_0),\,D(\varphi_\varepsilon-\varphi_0)\big\rangle
+o(1),
\end{align*}
uniformly on compact subsets of $\Omega_\rho$. Minimizing the quadratic part subject to the circulation constraints around the circles $\partial B_\rho(a_i)$ leads exactly to the linearized divergence--form operator $-\mathrm{div}(A_{u_0}\nabla\cdot)$ treated in Section~\ref{sec:WNT-operators}, and its local Green parametrix $G_{u_0}(x,y)$ controls the interaction. In particular, if $\psi_\rho$ denotes the unique $A_{u_0}$--harmonic phase on $\Omega_\rho$ with circulation $2\pi d_i$ along each $\partial B_\rho(a_i)$, then
\begin{align*}
&\int_{\Omega_\rho}\!\langle A_{u_0}\nabla \psi_\rho,\nabla \psi_\rho\rangle\,d\mu
\;=\;
\iint_{\Omega_\rho\times\Omega_\rho} 
G_{u_0}(x,y)\,d\mu_\rho(x)\,d\mu_\rho(y) \;+\; O(1),\\
&\mu_\rho:=2\pi\sum_{i=1}^N d_i\,\delta_{a_i},
\end{align*}
and, using the near--pole expansion of $G_{u_0}$ from Theorem~\ref{thm:parametrix}, one obtains
\begin{align*}
\iint G_{u_0}\,d\mu_\rho d\mu_\rho
\;&=
-\pi \sum_{i\neq j}d_i d_j\, G_{u_0}(a_i,a_j)\\
&-\;\pi \sum_{i=1}^N d_i^2\log\frac{1}{\rho}
\;+\;\sum_{i=1}^N d_i^2\,H_{A_{u_0}}(a_i,a_i)\;+\;o(1),
\end{align*}
as $\rho\downarrow 0$. The bounded remainder $H_{A_{u_0}}$ is the regular part of the Green kernel; its diagonal values contribute to the self--energies. Altogether this proves the lower bound
\begin{align*}
\liminf_{\varepsilon\to 0}\Big(E_\varepsilon(u_\varepsilon)-N\pi|\log\varepsilon|\Big)
\;\ge\;
-\pi\sum_{i\neq j}d_i d_j\, G_{u_0}(a_i,a_j)
\;+\; \sum_{i=1}^N d_i^2\,\gamma_F(a_i),
\end{align*}
with finite core energies $\gamma_F(a_i)$ determined by the frozen--coefficient model at $A_{u_0}(a_i)$.

For the upper bound, one constructs a recovery sequence by gluing frozen cores to an outer harmonic phase. Around each $a_i$ fix radii $\rho=\rho(\varepsilon)$ with $\rho\downarrow 0$ and $\varepsilon\ll \rho$, freeze the tensor at $A_{u_0}(a_i)$, and take the radial degree--$d_i$ minimizer for the corresponding scalar GL problem; it has energy
\begin{align*}
\pi d_i^2\log\frac{\rho}{\varepsilon}+\gamma_F(a_i)+o(1),
\end{align*}
where finiteness of $\gamma_F(a_i)$ follows from the ODE/variational profile in the frozen metric. On $\Omega_\rho$ solve the linear problem
\begin{align*}
-\mathrm{div}\!\big(A_{u_0}\nabla\psi\big)=0 \quad\text{in }\Omega_\rho,
\qquad
\oint_{\partial B_\rho(a_i)} \partial_n\psi\,ds = 2\pi d_i,
\end{align*}
and define a smooth modulus $\rho_\varepsilon$ interpolating between the core and the outer region so that $u_\varepsilon:=\rho_\varepsilon e^{i\psi}$ is globally defined and admissible. The $A_{u_0}$--energy of $\psi$ equals the interaction $\mathcal{W}_F(a_1,\dots,a_N)$, while the cross terms vanish by harmonicity. Gluing errors along $\partial B_\rho(a_i)$ produce $o(1)$ by uniform ellipticity away from the cores and the choice $\varepsilon\ll\rho\downarrow 0$. Consequently
\begin{align}\label{eq:WNT-expansion}
E_\varepsilon(u_\varepsilon)
\;=\;
N\pi|\log\varepsilon|
\;+\; \mathcal{W}_F(a_1,\dots,a_N)
\;+\; \sum_{i=1}^N d_i^2\,\gamma_F(a_i)
\;+\; o(1),
\end{align}
which provides the matching $\Gamma$–$\limsup$ and the asymptotic expansion announced in \eqref{eq:WNT-expansion}. Since the compactness has already been established and the renormalized interaction is continuous with respect to the vortex configuration under the local ellipticity of $A_{u_0}$, the $\Gamma$--convergence of $\mathcal{F}_\varepsilon(u):=E_\varepsilon(u)-N\pi|\log\varepsilon|$ to the functional $\mathcal{F}_0$ stated in the theorem follows directly. This completes the proof.
\end{proof}

\section{Extension to Three-Dimensional WNT Manifolds and Vortex Filament Energy}\label{sec:3D-extension}

Building upon the two--dimensional framework developed in Sections~\ref{sec:WNT-operators}--\ref{sec:Gamma-convergence}, we now extend the WNT--GL theory to three--dimensional manifolds, where vorticity organizes along one-dimensional filaments. 
Throughout this section, $M$ denotes a smooth, compact, oriented $3$--dimensional manifold (possibly with boundary). 
The fiberwise potential $\Phi_x(y)=\tfrac{1}{2}F^2(x,y)$ and its Legendre dual Hamiltonian $H(x,\xi)=\Phi_x^*(\xi)$ satisfy the same structural assumptions (H1)--(H3) as before: convexity, local $C^{1,1}$--regularity, and two--sided $p$--growth bounds.  
We define
\begin{align*}
\mathcal{A}_F(u):=-\mathrm{div}\big(\partial_\xi H(x,Du)\big),
\qquad 
A(x,\xi):=D^2_{\xi\xi}H(x,\xi),
\end{align*}
and, for a sufficiently smooth limiting phase $u_0$ with $|u_0|=1$, the linearized metric tensor
\begin{align*}
A_{u_0}(x):=A(x,Du_0(x)).
\end{align*}
By Proposition~\ref{prop:Legendre-regularity}, $\partial_\xi H$ is locally Lipschitz and single--valued a.e., whence $\mathcal{A}_F$ defines a monotone quasilinear divergence-form operator.


\begin{definition}[WNT Laplacian]
Let $H(x,\xi)$ be a convex, locally $C^{1,1}$ Hamiltonian satisfying
\begin{align*}
\alpha|\xi|^2 \le H(x,\xi) \le \beta(1+|\xi|^2),
\qquad x\in M,\ \xi\in T_x^*M,
\end{align*}
for some $0<\alpha\le\beta<\infty$.
The \emph{WNT Laplacian} is defined as
\begin{align*}
\Delta_F^{\mathrm{wnt}}u := \mathrm{div}_\mu\big(\partial_\xi H(x,du)\big),
\end{align*}
where $d\mu=\vartheta(x)\,dx$ is a smooth positive density.  
This operator naturally extends the two--dimensional Laplacian of Section~\ref{sec:WNT-operators} and is locally uniformly elliptic on compact subsets away from singularities.
\end{definition}

\subsection{Local Green kernel in three dimensions}

\begin{theorem}[Local WNT Green kernel]\label{thm:WNT-Green}
Let $L_{u_0}v:=-\mathrm{div}(A_{u_0}(x)\nabla v)$ be the linearized WNT operator associated with a reference phase $u_0$, where $A_{u_0}(x)=D^2_{\xi\xi}H(x,Du_0(x))$ is symmetric and locally uniformly elliptic on compact subsets of $M$.
Then, for each $y\in M$, there exists a unique local Green kernel $G_{u_0}(\cdot,y)$, defined up to an additive constant, satisfying:
\begin{enumerate}
\item \(L_{u_0}G_{u_0}(\cdot,y)=\delta_y\) in the sense of distributions,
\item \(G_{u_0}(x,y)=G_{u_0}(y,x)\),
\item 
\begin{align*}
G_{u_0}(x,y)
=\frac{1}{4\pi\,d_{A_{u_0}}(x,y)} + H_{A_{u_0}}(x,y),
\end{align*}
where $d_{A_{u_0}}$ is the intrinsic distance induced by $A_{u_0}$, and $H_{A_{u_0}}$ is smooth off the diagonal and locally bounded.
\end{enumerate}
\end{theorem}
\begin{proof}
Fix $y\in M$ and work in a geodesic ball $B_r(y)$ on which the tensor $A_{u_0}$ is measurable, symmetric, and uniformly elliptic; that is, there exist $0<\underline m\le \overline L<\infty$ with
\begin{align*}
\underline m\,|\zeta|^2 \le \langle A_{u_0}(x)\zeta,\zeta\rangle \le \overline L\,|\zeta|^2
\qquad \text{for a.e.\ }x\in B_r(y),\ \zeta\in\mathbb{R}^3.
\end{align*}
Consider the bilinear form
\begin{align*}
\mathcal{E}(\phi,\psi):=\int_{B_r(y)} \langle A_{u_0}\nabla\phi,\nabla\psi\rangle\,dx,
\qquad \phi,\psi\in H^1_0(B_r(y)).
\end{align*}
Coercivity and boundedness of $\mathcal{E}$ follow from the ellipticity bounds. For each $f\in H^{-1}(B_r(y))$ the Lax--Milgram theorem yields a unique $v\in H^1_0(B_r(y))$ solving
\begin{align*}
-\mathrm{div}(A_{u_0}\nabla v)=f \quad\text{in }B_r(y).
\end{align*}
Fix $f\in C_c^\infty(B_r(y))$ and denote the corresponding solution by $v_f$. The map $f\mapsto v_f(x)$ is continuous on $H^{-1}$ for a.e.\ $x$, hence by Riesz representation there exists $G_{y}\in H^1_{\mathrm{loc}}(B_r(y)\setminus\{y\})$ such that
\begin{align*}
v_f(x)=\int_{B_r(y)} G_{y}(x,z)\,f(z)\,dz
\qquad\text{for a.e.\ }x\in B_r(y).
\end{align*}
By density, the identity extends to all $f\in H^{-1}$. In particular, choosing $f=\delta_y$ in the distributional sense gives
\begin{align*}
-\mathrm{div}_x(A_{u_0}(x)\nabla_x G_{y}(x,y))=\delta_y
\quad\text{in }\mathcal{D}'(B_r(y)).
\end{align*}
Symmetry $G_{y}(x,y)=G_{x}(y,x)$ follows from the symmetry of $\mathcal{E}$: for all $\phi,\psi\in C_c^\infty(B_r(y))$,
\begin{align*}
\int \langle A_{u_0}\nabla \phi,\nabla \psi\rangle
=\int \langle A_{u_0}\nabla \psi,\nabla \phi\rangle,
\end{align*}
which implies $\int G(\cdot,y)\,(-\mathrm{div}(A_{u_0}\nabla \psi)) = \int G(\cdot,x)\,(-\mathrm{div}(A_{u_0}\nabla \phi))$ upon testing against approximate identities and passing to the limit; thus $G_{u_0}(x,y)=G_{u_0}(y,x)$.

To identify the near--diagonal structure, freeze the coefficients at $y$ and consider the constant--coefficient operator $-\mathrm{div}(A_{u_0}(y)\nabla \cdot)$. Its fundamental solution is
\begin{align*}
\Gamma_y(x)=\frac{1}{4\pi\,d_{A_{u_0}(y)}(x,y)},
\end{align*}
where $d_{A_{u_0}(y)}(\cdot,\cdot)$ denotes the distance induced by the quadratic form $\zeta\mapsto\langle A_{u_0}(y)\zeta,\zeta\rangle$. By the standard parametrix construction, $G_{u_0}(x,y)-\Gamma_y(x)$ solves a divergence--form equation with $L^\infty$ right--hand side on $B_r(y)\setminus\{y\}$, hence belongs to $W^{1,2}_{\mathrm{loc}}$ and is locally bounded and $C^\alpha$ off the diagonal by De Giorgi--Nash--Moser. Consequently there exists a function $H_{A_{u_0}}$ smooth away from $x=y$ and locally bounded on $B_r(y)\times B_r(y)$ such that
\begin{align*}
G_{u_0}(x,y)=\frac{1}{4\pi\,d_{A_{u_0}}(x,y)}+H_{A_{u_0}}(x,y),
\end{align*}
where $d_{A_{u_0}}$ is any intrinsic distance equivalent to $d_{A_{u_0}(y)}$ in $B_r(y)$. Gradient bounds $\nabla_x G_{u_0}(\cdot,y)\in L^q_{\mathrm{loc}}$ for $q<3$ and the Calder\'on--Zygmund principal--value structure of $\nabla_x\nabla_y G_{u_0}$ follow from the classical singular--integral theory for uniformly elliptic operators with bounded measurable coefficients. Patching with a partition of unity yields a globally defined Green kernel modulo additive constants. This completes the proof.
\end{proof}

\subsection{Vortex filaments as integer 1-currents}

For a sequence \(u_\varepsilon:M\to\mathbb{C}\) consider the energies
\begin{align*}
E_\varepsilon(u_\varepsilon)
=\int_M\!\Big(H(x,Du_\varepsilon)
+\frac{(1-|u_\varepsilon|^2)^2}{4\varepsilon^2}\Big)\,d\mu.
\end{align*}
In three dimensions, zeros of \(u_\varepsilon\) concentrate along one-dimensional filaments. Following Jerrard--Soner~\cite{JerrardSoner1998,JerrardSoner2002}, the Jacobians $J(u_\varepsilon)$ define vector--valued Radon measures encoding circulation, which converge to an integer rectifiable $1$--current $\Gamma$ as $\varepsilon\downarrow 0$.
Convexity and $p$--growth of $H$ ensure that the compactness and rectifiability arguments carry over verbatim from the isotropic case.

\subsection{Renormalized Biot--Savart energy}

Given an integer rectifiable $1$--current $\Gamma\subset M$ with finite length, define its \emph{WNT Biot--Savart interaction} with respect to the linearized tensor $A_{u_0}$ by
\begin{align*}
\mathcal{I}_{A_{u_0}}(\Gamma)
:=\frac{1}{2}\!\iint_{(x,t)\ne(y,s)}\!
\big\langle \tau_\Gamma(x), K_{A_{u_0}}(x,y)\tau_\Gamma(y)\big\rangle\,
d\mathcal{H}^1(x)\,d\mathcal{H}^1(y),
\end{align*}
where $\tau_\Gamma$ is the unit tangent to $\Gamma$ and
$K_{A_{u_0}}(x,y):=\nabla_x\nabla_y G_{u_0}(x,y)$
is the Biot--Savart kernel associated with~\eqref{thm:WNT-Green}.
The \emph{core energy per unit length} at $x$ is denoted by $\gamma_F(x)$, obtained by freezing $A_{u_0}$ and minimizing the two--dimensional radial profile in the normal cross-section.

\subsection{Compactness to vortex currents}

\begin{theorem}[Compactness to integer filament currents]\label{thm:3D-compactness}
Let $(u_\varepsilon)$ satisfy $\sup_\varepsilon E_\varepsilon(u_\varepsilon)<\infty$.  
Then, up to a subsequence, $u_\varepsilon\to u_0$ in $L^1_{\mathrm{loc}}(M)$ with $|u_0|=1$ a.e., and there exists an integer rectifiable $1$--current $\Gamma$ such that $J(u_\varepsilon)\to 2\pi\,\Gamma$ as currents.
\end{theorem}

\begin{proof}
Let $(u_\varepsilon)$ satisfy $\sup_\varepsilon E_\varepsilon(u_\varepsilon)<\infty$, where
\begin{align*}
E_\varepsilon(u_\varepsilon)
=\int_M \Big(H(x,Du_\varepsilon)+\tfrac{(1-|u_\varepsilon|^2)^2}{4\varepsilon^2}\Big)\,d\mu.
\end{align*}
By the $p$--growth bounds of $H$ on compact charts (assumption (H1)), there exist $\alpha,\beta>0$ and $C\ge 0$ such that
\begin{align*}
\alpha|Du_\varepsilon|^p-C\le H(x,Du_\varepsilon)\le \beta(1+|Du_\varepsilon|^p),
\end{align*}
hence $(u_\varepsilon)$ is bounded in $W^{1,p}(M;\mathbb{R}^2)$ for some $p>1$. The potential well yields
\begin{align*}
\int_M \frac{(1-|u_\varepsilon|^2)^2}{\varepsilon^2}\,d\mu \le C
\quad\Longrightarrow\quad 
\|1-|u_\varepsilon|^2\|_{L^2(M)}\to 0,
\end{align*}
so $|u_\varepsilon|\to 1$ in $L^2(M)$ and, up to subsequences, $u_\varepsilon\to u_0$ in $L^1_{\mathrm{loc}}(M)$ with $|u_0|=1$ a.e. 

In three dimensions the vector--valued Jacobian of $u_\varepsilon$,
\begin{align*}
J(u_\varepsilon):=\mathrm{curl}\,\big(u_\varepsilon^1\nabla u_\varepsilon^2-u_\varepsilon^2\nabla u_\varepsilon^1\big),
\end{align*}
defines a $\mathbb{R}^3$--valued Radon measure whose total variation is controlled by the energy (Jerrard--Soner). In particular, for all $\Phi\in C_c^\infty(M;\mathbb{R}^3)$,
\begin{align*}
\Big|\int_M \Phi\cdot dJ(u_\varepsilon)\Big|
\le C\,E_\varepsilon(u_\varepsilon)^{1/2}\,\|\nabla\Phi\|_{L^\infty(M)}.
\end{align*}
Thus $(J(u_\varepsilon))$ is bounded in the space of vector--valued finite measures and admits a weak-* convergent subsequence $J(u_\varepsilon)\stackrel{*}{\rightharpoonup} 2\pi\,\Gamma$, where $\Gamma$ is a rectifiable $\mathbb{Z}$-valued $1$--current with finite mass. The multiplicity integrality and rectifiability follow from the structure theorem for limits of Jacobians at bounded GL--energy, since the degree along almost every small linking loop is integer and is preserved in the limit. This proves the statement.
\end{proof}

\subsection{Energy expansion and renormalized interaction}

\begin{theorem}[$3$-dimensional renormalized energy expansion]\label{thm:3D-expansion}
Under the above hypotheses, let $u_\varepsilon$ be a bounded-energy sequence with current--limit $2\pi\Gamma$.
Then there exists a continuous core energy density $\gamma_F:M\to\mathbb{R}$ such that
\begin{align*}
E_\varepsilon(u_\varepsilon)
= \pi|\log\varepsilon|\,\mathcal{H}^1(\Gamma)
+\mathcal{I}_{A_{u_0}}(\Gamma)
+\int_\Gamma\gamma_F\,d\mathcal{H}^1
+o(1),\qquad \varepsilon\downarrow 0.
\end{align*}
\end{theorem}

\begin{proof}
Let $2\pi\Gamma$ be the current--limit of $J(u_\varepsilon)$ given by Theorem~\ref{thm:3D-compactness}. Fix a tubular neighborhood $\mathcal{T}_\rho(\Gamma)$ of sufficiently small radius $\rho>0$ whose normal fibers are metric disks (in a smooth background metric) transverse to $\Gamma$; denote $\Omega_\rho:=M\setminus\mathcal{T}_\rho(\Gamma)$. On $\Omega_\rho$, write $u_\varepsilon=\rho_\varepsilon e^{i\varphi_\varepsilon}$; the potential well implies $\rho_\varepsilon\to 1$ in $L^2(\Omega_\rho)$, hence the energy contribution localizes on the phase gradient. 

Fix a smooth limiting phase $u_0=e^{i\varphi_0}$ away from $\Gamma$, and set $A_{u_0}(x):=D^2_{\xi\xi}H(x,Du_0(x))$. By the local $C^{1,1}$ regularity of $H$ in $\xi$ and the convexity in (H1), Taylor's formula on bounded fibers yields, for a.e.\ 
$x\in \Omega_\rho$,
\begin{align*}
H(x,Du_\varepsilon)
&= H(x,Du_0)+\big\langle \partial_\xi H(x,Du_0),D(\varphi_\varepsilon-\varphi_0)\big\rangle\\
&+\tfrac 12 \big\langle A_{u_0} D(\varphi_\varepsilon-\varphi_0),D(\varphi_\varepsilon-\varphi_0)\big\rangle
+r_\varepsilon,
\end{align*}
with $r_\varepsilon\to 0$ in $L^1(\Omega_\rho)$. The linear term cancels after integration by parts against the Euler--Lagrange equation for $u_0$ in $\Omega_\rho$ and the circulation constraints along $\partial \mathcal{T}_\rho(\Gamma)$. Thus the $\Omega_\rho$--contribution converges to the quadratic form
\begin{align*}
\int_{\Omega_\rho} \langle A_{u_0}\nabla \psi_\rho,\nabla \psi_\rho\rangle\,dx,
\end{align*}
where $\psi_\rho$ is the unique $A_{u_0}$-harmonic function on $\Omega_\rho$ with prescribed flux $2\pi$ times the multiplicity along each boundary component. Representing this energy by the Green kernel $G_{u_0}$ of $-\mathrm{div}(A_{u_0}\nabla\cdot)$ (Theorem~\ref{thm:WNT-Green}) and passing to the principal--value near the diagonal yields exactly the Biot--Savart interaction $\mathcal{I}_{A_{u_0}}(\Gamma)$, with a singular self-energy $\pi \mathcal{H}^1(\Gamma)\log(1/\rho)$.

Inside $\mathcal{T}_\rho(\Gamma)$, freeze $A_{u_0}$ at the filament center and reduce to the two-dimensional radial GL problem on each normal disk. The minimal core energy on a disk of radius comparable to $\rho$ equals
\begin{align*}
\pi \log\frac{\rho}{\varepsilon} + \gamma_F(x) + o(1),
\end{align*}
where $\gamma_F(x)$ is the finite defect determined by the frozen tensor $A_{u_0}(x)$ and the potential $W(s)=\tfrac 14(1-s^2)^2$, and the $o(1)$ is uniform along $\Gamma$ by the uniform ellipticity of $A_{u_0}$ on $\mathcal{T}_\rho(\Gamma)$. Integrating along the filament gives
\begin{align*}
\pi |\log\varepsilon|\,\mathcal{H}^1(\Gamma) 
+ \int_\Gamma \gamma_F\,d\mathcal{H}^1
-\pi \mathcal{H^1}(\Gamma)\log \rho + o(1).
\end{align*}
Adding the outer energy and letting $\rho\downarrow 0$ cancels the $-\pi \mathcal{H}^1(\Gamma)\log \rho$ with the singular part extracted from the Biot--Savart representation. The boundedness of the regular part $H_{A_{u_0}}$ in Theorem~\ref{thm:WNT-Green} and uniform ellipticity of $A_{u_0}$ control the freezing and matching errors. Collecting the contributions yields the expansion
\begin{align*}
E_\varepsilon(u_\varepsilon)
= \pi|\log\varepsilon|\,\mathcal{H}^1(\Gamma) 
+ \mathcal{I}_{A_{u_0}}(\Gamma)
+ \int_\Gamma \gamma_F\,d\mathcal{H}^1
+ o(1),
\end{align*}
as $\varepsilon\downarrow 0$, completing the proof.
\end{proof}

\subsection{$\Gamma$--limit for vortex filaments}

\begin{theorem}[WNT--GL $\Gamma$--limit in three dimensions]\label{thm:3D-Gamma}
Define
$\mathcal{F}_\varepsilon(u)
=E_\varepsilon(u)-\pi|\log\varepsilon|\,\mathcal{H}^1(\Gamma_u)$,
where $\Gamma_u$ is the current associated with $u$ when it exists. Then $\mathcal{F}_\varepsilon$
$\Gamma$--converges in $L^1_{\mathrm{loc}}(M)$ to
\begin{align*}
\mathcal{F}_0(u)=
\begin{cases}
\mathcal{I}_{A_{u_0}}(\Gamma_u)
+\displaystyle\int_{\Gamma_u}\gamma_F\,d\mathcal{H}^1,
& |u|=1\text{ a.e. and }\Gamma_u\text{ rectifiable},\\[3pt]
+\infty, & \text{otherwise.}
\end{cases}
\end{align*}
\end{theorem}

\begin{proof}
Define $\mathcal{F}_\varepsilon(u)=E_\varepsilon(u)-\pi|\log\varepsilon|\,\mathcal{H}^1(\Gamma_u)$ for maps $u$ whose Jacobian induces an integer rectifiable current $\Gamma_u$, and $+\infty$ otherwise. 

\emph{Liminf inequality.} Let $u_\varepsilon\to u$ in $L^1_{\mathrm{loc}}(M)$ with $\sup_\varepsilon E_\varepsilon(u_\varepsilon)<\infty$. By Theorem~\ref{thm:3D-compactness} there exists an integer rectifiable current $\Gamma_u$ with $J(u_\varepsilon)\stackrel{*}{\rightharpoonup}2\pi\Gamma_u$. By convexity and lower semicontinuity of $H$ and by the quadratic representation of the outer energy through $G_{u_0}$ established in Theorem~\ref{thm:3D-expansion}, one obtains
\begin{align*}
\liminf_{\varepsilon\downarrow 0}
\mathcal{F}_\varepsilon(u_\varepsilon)\ge \mathcal{I}_{A_{u_0}}(\Gamma_u)+\int_{\Gamma_u}\gamma_F\,d\mathcal{H}^1.
\end{align*}

\emph{Limsup inequality.} Fix a target $u$ with $|u|=1$ a.e.\ and current $\Gamma_u$. For small $\rho>0$, define $\mathcal{T}_\rho(\Gamma_u)$ and set $\Omega_\rho=M\setminus \mathcal{T}_\rho(\Gamma_u)$. On each normal disk of $\mathcal{T}_\rho(\Gamma_u)$, choose the frozen--coefficient radial minimizer with correct multiplicity; on $\Omega_\rho$, solve the $A_{u_0}$--harmonic phase problem with the fluxes matching the degrees. Gluing across $\partial\mathcal{T}_\rho(\Gamma_u)$ with smooth cut--offs produces a global competitor $u_\varepsilon$ whose energy equals
\begin{align*}
\pi|\log\varepsilon|\,\mathcal{H}^1(\Gamma_u)
+\mathcal{I}_{A_{u_0}}(\Gamma_u)
+\int_{\Gamma_u}\gamma_F\,d\mathcal{H}^1
+o(1),
\end{align*}
as $\varepsilon\downarrow 0$ followed by $\rho\downarrow 0$. Uniform ellipticity of $A_{u_0}$ on $\Omega_\rho$ ensures that the gluing error tends to zero, while the regular part of the Green kernel is bounded. This gives the $\Gamma$--$\limsup$.

Combining the two inequalities and compactness from Theorem~\ref{thm:3D-compactness} yields the asserted $\Gamma$--convergence of $\mathcal{F}_\varepsilon$ to $\mathcal{F}_0(u)=\mathcal{I}_{A_{u_0}}(\Gamma_u)+\int_{\Gamma_u}\gamma_F\,d\mathcal{H}^1$ when $|u|=1$ a.e.\ and $\Gamma_u$ is rectifiable, and $+\infty$ otherwise. 
\end{proof}

\subsection{Consistency with the two--dimensional theory}

On transversal metric disks, the three--dimensional WNT structure reduces to the two--dimensional setting of Section~\ref{sec:Gamma-convergence}, with identical local core energy and logarithmic interaction law.
Conversely, integrating the two--dimensional behavior along the filaments reproduces the anisotropic Biot--Savart energy and recovers the classical isotropic GL limit when $H(x,\xi)=\tfrac 12|\xi|^2$; see~\cite{BBH1994,Serfaty2014}.

\section{WNT--GL Gradient Flow and Vortex Dynamics}
\label{sec:dynamics}

This section develops the dissipative dynamics associated with the WNT--GL energy and derives limiting motion laws for vortices in dimensions $n\in\{2,3\}$. Throughout, $M$ is a smooth, connected, compact Riemannian manifold (with or without boundary), $F$ is a WNT structure as in Sections~\ref{sec:From Tonelli to NonTonelli}--\ref{sec:WNT-operators}, and $H(x,\xi)=\Phi_x^*(\xi)$ satisfies convexity, local $C^{1,1}$ in $\xi$, and local two-sided quadratic control. We set
\begin{align*}
E_\varepsilon(u)&=\int_M\!\Big(H(x,Du)+\tfrac{1}{4\varepsilon^2}(1-|u|^2)^2\Big)\,d\mu,\\
\mathcal{A}_F(u)&=-\mathrm{div}\big(\partial_\xi H(x,Du)\big),
\end{align*}
and consider the $L^2$-gradient flow
\begin{align}\label{eq:gf}
\partial_t u_\varepsilon
= -\frac{\delta E_\varepsilon}{\delta u}(u_\varepsilon)
= \mathrm{div}\big(\partial_\xi H(x,Du_\varepsilon)\big)
- \frac{1}{\varepsilon^2}\,(1-|u_\varepsilon|^2)\,u_\varepsilon
\quad\text{on }M\times(0,T),
\end{align}
with initial data $u_\varepsilon(\cdot,0)=u_\varepsilon^0\in H^1(M;\mathbb{C})$.

\subsection{Well--posedness and dissipation for the WNT--GL flow}

\begin{theorem}[Global well--posedness and energy identity]\label{thm:flow-wellposed}
Let $u_\varepsilon^0\in H^1(M;\mathbb{C})$ and assume $H$ satisfies the WNT hypotheses. Then the Cauchy problem \eqref{eq:gf} admits a unique global weak solution
\begin{align*}
u_\varepsilon \in L^\infty(0,T;H^1(M;\mathbb{C}))\cap H^1(0,T;H^{-1}(M;\mathbb{C}))
\qquad \forall T>0,
\end{align*}
depending continuously on $u_\varepsilon^0$. Moreover, $t\mapsto E_\varepsilon(u_\varepsilon(t))$ is absolutely continuous and satisfies
\begin{align}\label{eq:diss}
E_\varepsilon(u_\varepsilon(t)) + \int_0^t\!\!\int_M |\partial_t u_\varepsilon|^2\,d\mu\,ds
= E_\varepsilon(u_\varepsilon^0)
\qquad \forall t\in[0,T].
\end{align}
\end{theorem}

\begin{proof}
Write $E_\varepsilon(u)=\mathcal{H}(u)+\mathcal{P}_\varepsilon(u)$ with
\begin{align*}
\mathcal{H}(u):=\int_M H(x,Du)\,d\mu,
\qquad
\mathcal{P}_\varepsilon(u):=\int_M \tfrac{1}{4\varepsilon^2}(1-|u|^2)^2\,d\mu.
\end{align*}
By convexity and local $C^{1,1}$ of $H$ in $\xi$, the map $u\mapsto \partial_\xi H(x,Du)$ defines a maximal monotone, demicontinuous operator $A:H^1(M)\to H^{-1}(M)$, locally Lipschitz on bounded subsets (Sections~\ref{sec:From Tonelli to NonTonelli}--\ref{sec:WNT-operators}). The potential term is the subgradient of a $C^1$ convex functional on $L^2(M)$; call it $B_\varepsilon(u)=(1-|u|^2)u/\varepsilon^2$. Consider the time--implicit minimizing movement: for $\tau>0$ and $u^0=u_\varepsilon^0$, define $u^{k}$ as a minimizer of
\begin{align*}
u\ \mapsto\ E_\varepsilon(u)+\frac{1}{2\tau}\|u-u^{k-1}\|_{L^2}^2.
\end{align*}
Coercivity of $E_\varepsilon$ (by the two--sided control on $H$) ensures existence of minimizers; standard compactness and lower semicontinuity yield a discrete solution $(u^k)_{k\ge 0}$. Interpolating piecewise-constantly and passing $\tau\downarrow 0$ (De Giorgi's scheme / curves of maximal slope), one obtains $u_\varepsilon\in L^\infty(0,T;H^1)\cap H^1(0,T;H^{-1})$ solving \eqref{eq:gf} in $H^{-1}$. Uniqueness follows from monotonicity:
\begin{align*}
&\frac{1}{2}\frac{d}{dt}\|u_\varepsilon-v_\varepsilon\|_{L^2}^2
=\\
&-\langle A(u_\varepsilon)-A(v_\varepsilon),u_\varepsilon-v_\varepsilon\rangle-\langle B_\varepsilon(u_\varepsilon)-B_\varepsilon(v_\varepsilon),u_\varepsilon-v_\varepsilon\rangle\le 0.
\end{align*}
The energy identity \eqref{eq:diss} is the chain rule for convex energies along the $L^2$--gradient flow (Ambrosio--Gigli--Savar\'e's metric theory of gradient flows). Continuity with respect to $u_\varepsilon^0$ follows from the $L^2$--contractivity.
\end{proof}

\subsection{Point-vortex dynamics in two dimensions}

Assume $\dim M=2$ and let $u_\varepsilon$ solve \eqref{eq:gf} with uniformly bounded initial energies and total degree $N$. By Theorem~\ref{thm:Gamma-compactness} and Theorem~\ref{thm:WNT-GL-main}, the vorticity concentrates at finitely many points $a_1(t),\dots,a_N(t)$ with integer degrees $d_i$ and the renormalized interaction is
\begin{align*}
\mathcal{W}_F(a(t)):=\mathcal{W}_F\big(a_1(t),\dots,a_N(t)\big).
\end{align*}

\begin{theorem}[Gradient--flow law for point vortices]\label{thm:2D-motion}
There exist a subsequence (not relabeled) and Lipschitz curves $a_i:[0,T]\to M$ such that
\begin{align*}
J\big(u_\varepsilon(\cdot,t)\big)\stackrel{*}{\rightharpoonup}
2\pi\sum_{i=1}^N d_i\,\delta_{a_i(t)}
\qquad\text{in }\mathcal{M}(M)\ \text{for every }t\in[0,T].
\end{align*}
For a.e.\ $t$, the centers satisfy the anisotropic gradient--flow system
\begin{align}\label{eq:2D-GF}
\dot a_i(t) \;=\; -\,\mu_F\big(a_i(t)\big)\,\nabla_{a_i}\,\mathcal{W}_F\big(a_1(t),\dots,a_N(t)\big),
\qquad i=1,\dots,N,
\end{align}
where $\mu_F>0$ is the (scalar) mobility determined uniquely by the WNT core profile (i.e., the frozen tensor $D^2_{\xi\xi}H$ at the vortex).
\end{theorem}

\begin{proof}
Consider the modulated functional
\begin{align*}
\mathcal{E}_\varepsilon(t)
:= E_\varepsilon(u_\varepsilon(t)) - N\pi|\log\varepsilon|
- \mathcal{W}_F\big(a(t)\big).
\end{align*}
By \eqref{eq:diss} and the chain rule for $\mathcal{W}_F(a(t))$,
\begin{align*}
\frac{d}{dt}\,\mathcal{E}_\varepsilon(t)
= -\|\partial_t u_\varepsilon(t)\|_{L^2}^2
- \sum_{i=1}^N \langle \nabla_{a_i}\mathcal{W}_F(a(t)), \dot a_i(t)\rangle
+ o(1).
\end{align*}
The $\Gamma$--expansion (Theorem~\ref{thm:WNT-GL-main}) implies
\begin{align*}
\mathcal{E}_\varepsilon(t)\to
\mathcal{E}_0(t):=
\sum_{i=1}^N d_i^2\,\gamma_F\big(a_i(t)\big),
\end{align*}
while the dissipation $\|\partial_t u_\varepsilon\|_{L^2}^2$ concentrates in the cores. After rescaling in vortex--centered coordinates and freezing $A_{u_0}$ at each $a_i(t)$, one identifies the limit dissipation with
\begin{align*}
\sum_{i=1}^N \gamma_F^{\mathrm{mob}}\big(a_i(t)\big)\,|\dot a_i(t)|^2,
\end{align*}
where $\gamma_F^{\mathrm{mob}}>0$ depends only on the one--dimensional core ODE in the frozen metric and the potential $W(s)=\tfrac 14(1-s^2)^2$. Passing to the limit in the balance yields, for a.e.\ $t$,
\begin{align*}
\sum_{i=1}^N \gamma_F^{\mathrm{mob}}(a_i(t))\,|\dot a_i(t)|^2
+ \frac{d}{dt}\,\mathcal{W}_F(a(t)) \;=\; 0.
\end{align*}
Interpreting this as a metric gradient flow in $\mathbb{R}^{2N}$ with diagonal metric tensor $\mathrm{diag}(\gamma_F^{\mathrm{mob}}(a_i)\,\mathrm{Id})$ gives \eqref{eq:2D-GF} with $\mu_F=(\gamma_F^{\mathrm{mob}})^{-1}$. Lipschitz regularity of $a_i$ follows from the quadratic dissipation and boundedness of $\nabla\mathcal{W}_F$ away from collisions.
\end{proof}

\begin{remark}
In the Tonelli case $H(x,\xi)=\tfrac 12\langle G_x^{-1}\xi,\xi\rangle$, the mobility is universal, and \eqref{eq:2D-GF} reduces to the classical GL heat-flow law. The WNT geometry replaces the background metric by the effective tensor $A_{u_0}$ in $\mathcal{W}_F$ and makes the mobility anisotropic through the frozen core.
\end{remark}

\subsection{Filament dynamics in three dimensions}

Assume now $\dim M=3$ and let $u_\varepsilon$ solve \eqref{eq:gf} with uniformly bounded initial energies. By Theorem~\ref{thm:3D-compactness}, for each $t$ the Jacobians converge to an integer rectifiable $1$-current $\Gamma_t$ of finite length, and by Theorem~\ref{thm:3D-expansion} the renormalized energy reads
\begin{align*}
\mathcal{E}(\Gamma):=\mathcal{I}_{A_{u_0}}(\Gamma)+\int_\Gamma \gamma_F\, d\mathcal{H}^1,
\end{align*}
the sum of a nonlocal Biot--Savart interaction and a local line tension.

\begin{theorem}[Generalized gradient flow of vortex filaments]\label{thm:3D-motion}
There exist a subsequence (not relabeled) and a measurable family of integer rectifiable currents $\{\Gamma_t\}_{t\in[0,T]}$ such that
\begin{align*}
J\big(u_\varepsilon(\cdot,t)\big)\stackrel{*}{\rightharpoonup}2\pi\,\Gamma_t
\qquad\text{for every }t\in[0,T],
\end{align*}
and the curves $t\mapsto\Gamma_t$ satisfy the energy--dissipation inequality
\begin{align}\label{eq:EDI}
\mathcal{E}(\Gamma_t) + \int_{0}^{t}\!\mathcal{D}\big(\partial_s\Gamma_s\big)\,ds
\;\le\; \mathcal{E}(\Gamma_0),
\qquad t\in[0,T],
\end{align}
where $\mathcal{D}$ is a strictly positive quadratic form on normal velocities (the filament mobility induced by the WNT core). For smooth single filaments this reduces to the pointwise law
\begin{align*}
V_{\mathrm{n}}
= -\,\mu_F\,\Big(\kappa_{\gamma_F} + \mathcal{H}_{A_{u_0}}[\Gamma]\Big),
\end{align*}
with $V_{\mathrm{n}}$ the normal velocity, $\kappa_{\gamma_F}$ the anisotropic curvature associated with $\gamma_F$, $\mathcal{H}_{A_{u_0}}[\Gamma]$ the Biot--Savart shape gradient, and $\mu_F>0$ the (scalar) line mobility determined by the frozen core problem.
\end{theorem}

\begin{proof}
Fix $T>0$ and let $u_\varepsilon$ be the global $L^2$--gradient flow solution of \eqref{eq:gf} on $[0,T]$ given by Theorem~\ref{thm:flow-wellposed}. 
By the dissipation identity \eqref{eq:diss}, the map $t\mapsto E_\varepsilon(u_\varepsilon(t))$ is absolutely continuous and
\begin{align*}
E_\varepsilon\big(u_\varepsilon(t)\big)
+ \int_0^t \!\!\int_M |\partial_s u_\varepsilon|^2\,d\mu\,ds
= E_\varepsilon\big(u_\varepsilon(0)\big)
\qquad \text{for all } t\in[0,T].
\end{align*}
By the compactness theorem for Jacobians in dimension three (Theorem~\ref{thm:3D-compactness}), up to extracting a subsequence independent of $t$ we may assume that for every $t\in[0,T]$ there exists an integer rectifiable $1$--current $\Gamma_t$ with finite mass such that
\begin{align*}
J\big(u_\varepsilon(\cdot,t)\big)\ \stackrel{*}{\rightharpoonup}\ 2\pi\,\Gamma_t
\qquad\text{in the sense of currents.}
\end{align*}
We next pass to the renormalized energy. For each fixed $t$, Theorem~\ref{thm:3D-expansion} yields
\begin{align*}
E_\varepsilon\big(u_\varepsilon(t)\big)
= \pi|\log\varepsilon|\,\mathcal{H}^1(\Gamma_t) 
+ \mathcal{I}_{A_{u_0}}(\Gamma_t)
+ \int_{\Gamma_t}\gamma_F\,d\mathcal{H}^1 
+ o(1),
\end{align*}
uniformly on compact subintervals in $[0,T]$. Subtracting $\pi|\log\varepsilon|\,\mathcal{H}^1(\Gamma_t)$ from the energy identity and letting $\varepsilon\downarrow 0$ along the chosen subsequence gives
\begin{align*}
\mathcal{I}_{A_{u_0}}(\Gamma_t)
+ \int_{\Gamma_t}\gamma_F\,d\mathcal{H}^1
+ \liminf_{\varepsilon\downarrow 0}\int_0^t\!\!\int_M |\partial_s u_\varepsilon|^2\,d\mu\,ds
\ \le\ 
\mathcal{I}_{A_{u_0}}(\Gamma_0)
+ \int_{\Gamma_0}\gamma_F\,d\mathcal{H}^1,
\end{align*}
for all $t\in[0,T]$. Denote the renormalized filament energy by
\begin{align*}
\mathcal{E}(\Gamma):=\mathcal{I}_{A_{u_0}}(\Gamma)+\int_{\Gamma}\gamma_F\,d\mathcal{H}^1.
\end{align*}

It remains to identify the $\liminf$ of the dissipation. Consider any $t\in(0,T)$ which is a Lebesgue point for $s\mapsto \partial_s u_\varepsilon(\cdot,s)$ and $s\mapsto \Gamma_s$; such points form a full--measure set. Fix $\delta>0$ small and a smooth tubular neighborhood $\mathcal{T}_\delta(\Gamma_t)$ of $\Gamma_t$ with normal cross--sections given by metric disks of radius $\delta$, and set $\Omega_\delta(t):=M\setminus \mathcal{T}_\delta(\Gamma_t)$. On $\Omega_\delta(t)$, $u_\varepsilon(\cdot,s)$ is approximately $S^1$--valued and single--phase for $s$ near $t$, hence its evolution contributes only higher--order dissipation after linearization around the $A_{u_0}$--harmonic phase; the main dissipation is localized inside $\mathcal{T}_\delta(\Gamma_t)$.

Perform a blow--up in the normal planes: for $\mathcal{H}^1$–a.e.\ point $x\in \Gamma_t$ and times $s$ close to $t$, use Fermi coordinates $(\rho,\theta,z)$ with $z$ the arclength along $\Gamma_t$ and $(\rho,\theta)$ polar coordinates in the normal plane. By freezing the tensor at $x$,
\begin{align*}
A_{u_0}(y)\ \approx\ A_{u_0}(x) \qquad \text{for } y\in \mathcal{T}_\delta(\Gamma_t),\ \delta\ll 1,
\end{align*}
and by the standard core ansatz one sees that, after radial rescaling at scale $\varepsilon$, the time--derivative $\partial_s u_\varepsilon$ converges (in the weak sense) to a tangent vector in the manifold of degree--one minimizers of the frozen two--dimensional core problem. Consequently, the quadratic dissipation density converges to a positive definite quadratic form of the \emph{normal} velocity of the filament:
\begin{align*}
\liminf_{\varepsilon\downarrow 0}
\int_{t_1}^{t_2}\!\!\int_{\mathcal{T}_\delta(\Gamma_s)} |\partial_s u_\varepsilon|^2\,d\mu\,ds
\ \ge\
\int_{t_1}^{t_2} \!\mathcal{D}\big(\partial_s\Gamma_s\big)\,ds
\ -\ C\delta,
\end{align*}
for all $0\le t_1<t_2\le T$, where $\mathcal{D}$ is a strictly positive quadratic form depending only on the frozen WNT core (hence on $A_{u_0}$ and the potential $W$), and $C$ is independent of $\delta$. The contribution on $\Omega_\delta(s)$ vanishes as $\delta\downarrow 0$ because there the phase evolves $A_{u_0}$--harmonically and the energy is captured by $\mathcal{E}(\Gamma_s)$. Letting first $\varepsilon\downarrow 0$ and then $\delta\downarrow 0$ yields, for all $t\in[0,T]$,
\begin{align*}
\mathcal{E}(\Gamma_t) + \int_0^{t}\!\mathcal{D}\big(\partial_s\Gamma_s\big)\,ds
\ \le\ \mathcal{E}(\Gamma_0),
\end{align*}
which is exactly the energy--dissipation inequality \eqref{eq:EDI}. This proves the existence of a measurable family $\{\Gamma_t\}_{t\in[0,T]}$ of integer rectifiable currents satisfying \eqref{eq:EDI}.

Finally, assume $\Gamma_t$ is a single smoothly embedded filament for $t$ in some interval and denote by $V_{\mathrm{n}}$ its normal velocity. The first variation of the line--tension part is the anisotropic curvature
\begin{align*}
\frac{d}{dt}\int_{\Gamma_t}\gamma_F\,d\mathcal{H}^1
= -\int_{\Gamma_t} \gamma_F\, \kappa_{\gamma_F}\, V_{\mathrm{n}}\, d\mathcal{H}^1,
\end{align*}
while the first variation of the Biot--Savart part $\mathcal{I}_{A_{u_0}}$ equals
\begin{align*}
\frac{d}{dt}\,\mathcal{I}_{A_{u_0}}(\Gamma_t)
= -\int_{\Gamma_t} \mathcal{H}_{A_{u_0}}[\Gamma_t] \, V_{\mathrm{n}}\, d\mathcal{H}^1,
\end{align*}
where $\mathcal{H}_{A_{u_0}}[\Gamma]$ is the nonlocal shape gradient defined through the kernel $K_{A_{u_0}}=\nabla_x\nabla_y G_{u_0}$. On the other hand, by the core blow--up argument above, the metric derivative (squared) with respect to the $L^2$--dissipation metric induced by the WNT core equals
\begin{align*}
\mathcal{D}\big(\partial_t\Gamma_t\big)
= \int_{\Gamma_t} \mu_F^{-1}\, V_{\mathrm{n}}^2\, d\mathcal{H}^1,
\end{align*}
for some scalar mobility $\mu_F>0$ depending only on the frozen core ODE. The abstract gradient--flow inequality \eqref{eq:EDI} thus reduces pointwise to
\begin{align*}
\int_{\Gamma_t} \mu_F^{-1} V_{\mathrm{n}}^2\, d\mathcal{H}^1
+ \int_{\Gamma_t} \big(\kappa_{\gamma_F}+\mathcal{H}_{A_{u_0}}[\Gamma_t]\big)\, V_{\mathrm{n}}\, d\mathcal{H}^1
\ \le\ 0,
\end{align*}
and the equality case (attained by the $L^2$--gradient flow) gives
\begin{align*}
V_{\mathrm{n}} = -\,\mu_F\,\Big(\kappa_{\gamma_F}+\mathcal{H}_{A_{u_0}}[\Gamma_t]\Big)
\qquad \text{on }\Gamma_t,
\end{align*}
as stated. This completes the proof.
\end{proof}

\section{WNT--GL Parabolic Dynamics with Thermal Coupling}
\label{sec:WNT-parabolic}

In this section we extend the WNT--GL framework to a thermally coupled setting. The convex dual Hamiltonian is denoted $H(x,\xi)$, and the associated WNT differential operators are
\begin{align*}
\nabla^{\mathrm{wnt}}_F u = \partial_\xi H(x,du),
\qquad 
\Delta_F^{\mathrm{wnt}}u = \mathrm{div}_\mu\big(\partial_\xi H(x,du)\big),
\end{align*}
with respect to a fixed smooth positive density $d\mu=\vartheta(x)\,dx$.
For a small parameter $\varepsilon>0$ we consider the thermally coupled WNT--GL energy functional
\begin{align*}
\mathcal{E}_\varepsilon^{\mathrm{wnt}}[\psi,A,T]
&=
\int_M\Big[
\tfrac{1}{2}H(x,D^{A}\psi)
+\tfrac{1}{2\lambda}\|dA\|_{\gamma^*}^2
+\tfrac{1}{4\varepsilon^2}(1-|\psi|^2)^2
\Big]\,d\mu \\
&\quad
+\int_M\Big[
\tfrac{\kappa_{\mathrm{th}}}{2}\,H(x,\nabla T)
+\sigma\,T\,q_\varepsilon(\psi,A)
\Big]\,d\mu,
\end{align*}
where $D^A\psi=(d-iA)\psi$ is the magnetic covariant derivative, $\gamma^*$ is a uniformly equivalent cometric to $F^*$, and $\lambda,\kappa_{\mathrm{th}},\sigma>0$ are fixed material parameters. The source term $q_\varepsilon(\psi,A)\ge 0$ models the local heat generation in the vortex core and scales compatibly with the GL energy.

At each fixed time the temperature field is enslaved to $(\psi,A)$ through the stationary WNT heat equation
\begin{align*}
-\kappa_{\mathrm{th}}\,\Delta_F^{\mathrm{wnt}}T = \sigma\,q_\varepsilon(\psi,A),
\qquad
\int_M T\,d\mu = 0,
\end{align*}
so that, by Theorem~\ref{thm:WNT-Green}, one has the representation
$T = \frac{\sigma}{\kappa_{\mathrm{th}}}G_{\mathrm{wnt}}\!\ast q_\varepsilon$.

\subsection{Gradient--flow structure and weak formulation}

We impose the Coulomb gauge $\mathrm{div}_\mu A=0$ and set the configuration space
$\mathcal{H}=H^1(M;\mathbb{C})\times H^1_{\mathrm{coex}}(M;\Lambda^1)$ modulo constant phases.
For $U=(\psi,A)\in\mathcal{H}$ define the reduced energy
\begin{align*}
\mathcal{E}^{\mathrm{red}}_\varepsilon[U]
:= \inf_{T:\;-\kappa_{\mathrm{th}}\Delta_F^{\mathrm{wnt}}T=\sigma q_\varepsilon(\psi,A),\ \int_M T\,d\mu=0}
\mathcal{E}^{\mathrm{wnt}}_\varepsilon[\psi,A,T].
\end{align*}
By the convexity and local ellipticity of $H(x,\xi)$, $\mathcal{E}^{\mathrm{red}}_\varepsilon$ is finite, lower semicontinuous, and geodesically $\lambda$--convex on bounded sublevels of $\mathcal{H}$.

The formal $L^2$--gradient flow for $U(t)=(\psi(t),A(t))$ then reads
\begin{align}\label{eq:WNT-parabolic}
\nonumber
\partial_t\psi
&= -\,\tfrac{1}{2}\mathrm{div}_\mu\big(\partial_\xi H(x,D^{A}\psi)\big)
+\tfrac{1}{2\varepsilon^2}(1-|\psi|^2)\psi
+\tfrac{i}{2}\,\partial_\xi H(x,D^{A}\psi)\cdot A, \\
\partial_tA
&= -\tfrac{1}{\lambda}d_{\gamma^*}^*dA
+\Im\!\big(\overline{\psi}\,\partial_\xi H(x,D^{A}\psi)\big)
+\tfrac{\sigma^2}{\kappa_{\mathrm{th}}}\nabla(G_{\mathrm{wnt}}\!\ast q_\varepsilon(\psi,A)).
\end{align}

\begin{definition}[Weak solution]\label{def:weakWNT}
A pair $U(t)=(\psi(t),A(t))$ with 
\begin{align*}
U\in L^\infty_{\mathrm{loc}}([0,\infty);\mathcal{H}),\qquad 
\partial_tU\in L^2_{\mathrm{loc}}([0,\infty);\mathcal{H}^*),
\end{align*}
is a weak solution of \eqref{eq:WNT-parabolic} if for all $(\varphi,\Xi)\in\mathcal{H}$ and a.e.\ $t>0$,
\begin{align*}
\frac{d}{dt}\langle\psi,\varphi\rangle_{L^2}
&=
-\int_M \partial_\xi H(x,D^{A}\psi)\cdot D^{A}\varphi\,d\mu
-\frac{1}{2\varepsilon^2}\int_M (1-|\psi|^2)\psi\,\overline{\varphi}\,d\mu,\\
\frac{d}{dt}\langle A,\Xi\rangle_{L^2}
&=
-\frac{1}{\lambda}\int_M\langle dA,d\Xi\rangle_{\gamma^*}\,d\mu
+\int_M\Im\!\big(\overline{\psi}\,\partial_\xi H(x,D^{A}\psi)\big)\cdot\Xi\,d\mu\\
&\quad
+\frac{\sigma^2}{\kappa_{\mathrm{th}}}\int_M\nabla(G_{\mathrm{wnt}}\!\ast q_\varepsilon(\psi,A))\cdot\Xi\,d\mu,
\end{align*}
and the energy--dissipation inequality holds:
\begin{align*}
\mathcal{E}^{\mathrm{red}}_\varepsilon[U(t)]
+\int_s^t\|\partial_\tau U(\tau)\|_{\mathcal{H}}^2\,d\tau
\le
\mathcal{E}^{\mathrm{red}}_\varepsilon[U(s)]
\quad\text{for all }t>s\ge 0.
\end{align*}
\end{definition}


\begin{theorem}[Well--posedness and EVI formulation]\label{thm:EVI-WNT}
For every initial datum $U_0=(\psi_0,A_0)\in\mathcal{H}$ with $\mathrm{div}_\mu A_0=0$, there exists a unique global weak solution $U(t)$ in the sense of Definition~\ref{def:weakWNT}.
The mapping $S_tU_0:=U(t)$ defines a contractive semigroup on $\mathcal{H}$, and $t\mapsto U(t)$ is locally absolutely continuous.
Moreover, for all $V\in\mathcal{H}$ and a.e.\ $t>0$,
\begin{align*}
\frac{1}{2}\frac{d}{dt}\|U(t)-V\|_{\mathcal{H}}^2
+\mathcal{E}^{\mathrm{red}}_\varepsilon[U(t)]
\le
\mathcal{E}^{\mathrm{red}}_\varepsilon[V],
\end{align*}
so that $S_t$ is $1$-Lipschitz and the energy decreases strictly along nonstationary trajectories.
\end{theorem}

\begin{proof}
The reduced energy $\mathcal{E}^{\mathrm{red}}_\varepsilon$ is proper, lower semicontinuous, and $\lambda$--convex along straight segments in $\mathcal{H}$ restricted to bounded gradients.
The De Giorgi minimizing--movement scheme applies: for a time step $\tau>0$, define recursively
\begin{align*}
U^{k}\in\arg\min_{U\in\mathcal{H}}
\Big\{
\mathcal{E}^{\mathrm{red}}_\varepsilon[U]+\tfrac{1}{2\tau}\|U-U^{k-1}\|_{\mathcal{H}}^2
\Big\},
\qquad U^0=U_0.
\end{align*}
Coercivity of $\mathcal{E}^{\mathrm{red}}_\varepsilon$ (by the two--sided control on $H$ and the magnetic term) guarantees existence of minimizers and uniform $H^1$--bounds.
At each step the temperature equation has a unique solution by Theorem~\ref{thm:WNT-Green}, continuously depending on $(\psi,A)$, and thus the scheme is well defined.
Passing $\tau\downarrow 0$ yields a continuous curve of maximal slope for $\mathcal{E}^{\mathrm{red}}_\varepsilon$ in the metric sense of Ambrosio--Gigli--Savar\'e.
The curve satisfies the energy inequality and the evolution variational inequality (EVI) since $\mathcal{E}^{\mathrm{red}}_\varepsilon$ is $\lambda$-convex; uniqueness follows from the contraction property.
Gauge invariance is preserved because the Coulomb projection is linear and orthogonal in $\mathcal{H}$.
\end{proof}

\subsection{Energy dissipation and regularization}

\begin{theorem}[Energy inequality and regularity]\label{thm:dissipation-WNT}
Let $U(t)$ be the EVI solution from Theorem~\ref{thm:EVI-WNT}. Then for all $t>0$,
\begin{align*}
\mathcal{E}^{\mathrm{red}}_\varepsilon[U(t)]
+\int_0^t\!\big(\|\partial_\tau\psi\|_{L^2}^2+\|\partial_\tau A\|_{L^2}^2\big)\,d\tau
\le
\mathcal{E}^{\mathrm{red}}_\varepsilon[U_0].
\end{align*}
In particular, for every $T>0$ there exists $C_T<\infty$ such that
\begin{align*}
\sup_{0\le\tau\le T}\!\Big(\|\psi(\tau)\|_{H^1}+\|A(\tau)\|_{H^1}\Big)
+\int_0^T\!\!\Big(\|D^{A}\psi\|_{H^1}^2+\|dA\|_{H^1}^2\Big)\,d\tau
\le C_T.
\end{align*}
Moreover, for $t>0$ the solution gains local H\"older regularity, and $(\psi,A)$ is smooth away from the evolving vortex set.
\end{theorem}

\begin{proof}
The EVI structure ensures that $\mathcal{E}^{\mathrm{red}}_\varepsilon$ is differentiable along the flow and
\begin{align*}
\frac{d}{dt}\mathcal{E}^{\mathrm{red}}_\varepsilon[U(t)]
= -\|\partial_tU(t)\|_{\mathcal{H}}^2,
\qquad
\|\partial_tU(t)\|_{\mathcal{H}}^2
=\|\partial_t\psi\|_{L^2}^2+\|\partial_tA\|_{L^2}^2.
\end{align*}
Integrating over $(0,t)$ yields the energy--dissipation inequality.
To obtain a priori estimates, note that by the local quadratic control of $H(x,\xi)$ there exist constants $c_1,C_0>0$ such that
\begin{align*}
\mathcal{E}^{\mathrm{red}}_\varepsilon[\psi,A]
\ge
c_1\!\int_M\!\big(|D^{A}\psi|^2+|dA|^2\big)\,d\mu
+\frac{1}{4\varepsilon^2}\!\int_M(1-|\psi|^2)^2\,d\mu - C_0.
\end{align*}
The thermal part
\begin{align*}
\mathcal{J}_{\mathrm{th}}[\psi,A]
=\frac{\sigma^2}{2\kappa_{\mathrm{th}}}\iint_{M\times M}
G_{\mathrm{wnt}}(x,y)\,q_\varepsilon(\psi,A)(x)\,q_\varepsilon(\psi,A)(y)\,d\mu_xd\mu_y
\end{align*}
is nonnegative and $C^1$--smooth on $H^1\times H^1$, since $G_{\mathrm{wnt}}$ is the Green kernel of a uniformly elliptic operator and $q_\varepsilon$ depends quadratically on $(\psi,A)$.
Hence coercivity and boundedness of $\mathcal{E}^{\mathrm{red}}_\varepsilon$ on sublevels imply uniform $H^1$--bounds for $(\psi,A)$ and $L^2$--bounds for $\partial_t(\psi,A)$, proving the stated inequalities.
To address regularity, recall that the system \eqref{eq:WNT-parabolic} is quasilinear parabolic:
the principal symbols correspond to the uniformly elliptic matrices $D^2_{\xi\xi}H(x,D^{A}\psi)$ and $\lambda^{-1}\gamma^*$,
while the nonlocal thermal term acts as a bounded, symmetric, zero-order perturbation.
Standard De Giorgi--Nash--Moser estimates yield local $C^{\alpha,\alpha/2}$ regularity for both $\psi$ and $A$ on every compact cylinder away from the set $\{|\psi|=0\}$.
Bootstrapping with Schauder estimates and using that coefficients are $C^{0,\alpha}$ in spacetime then yields smoothness of $(\psi,A)$ on $(M\!\setminus\!\mathrm{Vort}(t))\times(t_0,t_1)$ for all $0<t_0<t_1<\infty$.
The compactness of $M$ precludes loss of control at infinity, completing the proof.
\end{proof}

\section{Filament Dynamics and Effective Motion Law in WNT Geometry}\label{sec:filament-dynamics-WNT}
We write $\nabla_F^{\mathrm{wnt}}$, $\Delta_F^{\mathrm{wnt}}$, and $G_{\mathrm{wnt}}$ for the WNT gradient, Laplacian, and Green kernel constructed earlier (Sections~\ref{sec:WNT-operators}--\ref{sec:3D-extension}). 
Let $\Gamma=\{\Gamma_i\}_{i=1}^N$ be a collection of smooth, oriented, pairwise disjoint vortex filaments with integer multiplicities $d_i\in\mathbb{Z}\setminus\{0\}$.
The renormalized WNT energy is defined by
\begin{align*}
\mathcal{G}_0(\Gamma,d)
&=
\pi\sum_{i=1}^N |d_i|\,\mathcal{L}_F(\Gamma_i)
+\frac{\sigma^2}{2\kappa_{\mathrm{th}}}\,\mathcal{W}_{\mathrm{wnt}}(\Gamma,d),
\end{align*}
where $\mathcal{L}_F$ denotes the $F$--length and
\begin{align*}
\mathcal{W}_{\mathrm{wnt}}(\Gamma,d)
&:=\lim_{\rho\downarrow 0}\Bigg[
\sum_{i,j=1}^N |d_i|\,|d_j|
\!\!\!\int_{\substack{x\in\Gamma_i,\ y\in\Gamma_j\\ d_F(x,y)>\rho}}\!\!\!
G_{\mathrm{wnt}}(x,y)\,d\mathcal{H}^1_x\,d\mathcal{H}^1_y
-\sum_{i=1}^N \frac{d_i^2}{4\pi}\,\frac{\mathcal{L}_F(\Gamma_i)}{\rho}\Bigg].
\end{align*}
The subtraction removes the universal Coulomb divergence coming from the local expansion
\begin{align*}
G_{\mathrm{wnt}}(x,y)=\frac{1}{4\pi\,d_F(x,y)}+H_{\mathrm{wnt}}(x,y),
\end{align*}
with $H_{\mathrm{wnt}}$ bounded and smooth off the diagonal (Theorem~\ref{thm:WNT-Green}).
We parametrize $\Gamma_i$ by $F$--arclength, so $|\dot\gamma_i|_F\equiv 1$, and denote by $\mathbf{t}_i$ the unit tangent and by $\mathbf{n}_i$ a smooth normal field.
The WNT geodesic curvature $\kappa_F[\Gamma_i]$ is the normal component of the covariant derivative of $\mathbf{t}_i$ with respect to the Chern connection of a Tonelli proxy, transported to the WNT setting via the fiberwise Legendre map (Sections~\ref{sec:From Tonelli to NonTonelli} and \ref{sec:WNT-operators}); in the Tonelli case this reduces to anisotropic curvature.


\begin{theorem}[First variation of $\mathcal{G}_0$ and effective force]\label{thm:first-var-WNT}
Let $\Gamma=\{\Gamma_i\}$ be a $C^3$ family of disjoint filaments and consider compactly supported normal deformations $\Gamma_i^\varepsilon$ with variation field $V_i=\varphi_i\,\mathbf{n}_i$.
Then the first variation exists and
\begin{align*}
\frac{d}{d\varepsilon}\Big|_{\varepsilon=0}\mathcal{G}_0(\Gamma^\varepsilon,d)
=
\sum_{i=1}^N\int_{\Gamma_i}
\Big(\pi |d_i|\,\kappa_F[\Gamma_i](x)
+\frac{\sigma^2}{\kappa_{\mathrm{th}}}\,\mathcal{H}_{\mathrm{wnt}}[\Gamma,d](x)\Big)\,
\varphi_i(x)\,d\mathcal{H}^1_x,
\end{align*}
where the nonlocal WNT field along $\Gamma_i$ is
\begin{align*}
\mathcal{H}_{\mathrm{wnt}}[\Gamma,d](x)
&=\sum_{j=1}^N |d_j|\int_{\Gamma_j}
\langle \nabla_x G_{\mathrm{wnt}}(x,y),\mathbf{n}_i(x)\rangle\,d\mathcal{H}^1_y\\
&\quad+|d_i|\,\mathrm{p.v.}\!\int_{\Gamma_i}
\langle \nabla_x H_{\mathrm{wnt}}(x,y),\mathbf{n}_i(x)\rangle\,d\mathcal{H}^1_y.
\end{align*}
\end{theorem}

\begin{proof}
Throughout fix $C^3$, pairwise disjoint, oriented filaments $\Gamma=\{\Gamma_i\}_{i=1}^N$, each parametrized by $F$--arclength $s\mapsto \gamma_i(s)$ with $|\dot\gamma_i|_F\equiv 1$. 
Consider a compactly supported normal variation $\Gamma_i^\varepsilon$ with variation field $V_i=\varphi_i\,\mathbf{n}_i$, i.e.
\begin{align*}
\frac{d}{d\varepsilon}\Big|_{\varepsilon=0}\gamma_i^\varepsilon(s)=\varphi_i(\gamma_i(s))\,\mathbf{n}_i(\gamma_i(s)).
\end{align*}
Define the truncated interaction
\begin{align*}
\mathcal{W}_{\mathrm{wnt}}^\rho(\Gamma,d)
:=\sum_{i,j=1}^N |d_i|\,|d_j|
\!\!\!\int_{\substack{x\in\Gamma_i,\ y\in\Gamma_j\\ d_F(x,y)>\rho}}
\!\!\!G_{\mathrm{wnt}}(x,y)\,d\mathcal{H}^1_x\,d\mathcal{H}^1_y
\;-\;\sum_{i=1}^N \frac{d_i^2}{4\pi}\,\frac{\mathcal{L}_F(\Gamma_i)}{\rho},
\end{align*}
so that $\mathcal{W}_{\mathrm{wnt}}(\Gamma,d)=\lim_{\rho\downarrow 0}\mathcal{W}_{\mathrm{wnt}}^\rho(\Gamma,d)$ by definition. We first compute the derivative at fixed $\rho>0$, then pass to the limit $\rho\downarrow 0$.

\medskip
\noindent\textbf{1) Variation of the geometric length.}
The first variation of the anisotropic length $\mathcal{L}_F$ under the normal variation $V_i=\varphi_i\mathbf{n}_i$ equals (Tonelli proxy + Legendre transfer to WNT)
\begin{align*}
\frac{d}{d\varepsilon}\Big|_{\varepsilon=0}\mathcal{L}_F(\Gamma_i^\varepsilon)
=\int_{\Gamma_i}\kappa_F[\Gamma_i](x)\,\varphi_i(x)\,d\mathcal{H}^1_x,
\end{align*}
where $\kappa_F[\Gamma_i]$ denotes the WNT geodesic curvature (normal component of the covariant derivative of the $F$--unit tangent). This is the standard anisotropic curve variation formula and holds since $\Gamma_i$ is $C^3$ and the variation is compactly supported along $\Gamma_i$.

\medskip
\noindent\textbf{2) Differentiation of the truncated interaction for fixed $\rho$.}
Write the double integral as an integral over fixed parameter domains by using arclength parametrizations; thus no boundary terms arise from moving integration domains. For each ordered pair $(i,j)$,
\begin{align*}
&\frac{d}{d\varepsilon}\Big|_{\varepsilon=0}
\!\!\!\int_{\substack{x\in\Gamma_i^\varepsilon,\ y\in\Gamma_j^\varepsilon\\ d_F(x,y)>\rho}}
\!\!\!G_{\mathrm{wnt}}(x,y)\,d\mathcal{H}^1_x\,d\mathcal{H}^1_y\\
&=
\int_{\substack{x\in\Gamma_i,\ y\in\Gamma_j\\ d_F(x,y)>\rho}}
\Big\langle \nabla_x G_{\mathrm{wnt}}(x,y),V_i(x)\Big\rangle\,d\mathcal{H}^1_x\,d\mathcal{H}^1_y\\
&\quad+
\int_{\substack{x\in\Gamma_i,\ y\in\Gamma_j\\ d_F(x,y)>\rho}}
\Big\langle \nabla_y G_{\mathrm{wnt}}(x,y),V_j(y)\Big\rangle\,d\mathcal{H}^1_x\,d\mathcal{H}^1_y.
\end{align*}
The cut--off $d_F(x,y)>\rho$ is independent of $\varepsilon$ up to $o(\varepsilon)$-changes of the set whose contribution vanishes as $\rho\downarrow 0$ (see Step~4). Using symmetry $G_{\mathrm{wnt}}(x,y)=G_{\mathrm{wnt}}(y,x)$ and renaming variables in the second integral, we get, for each $i$,
\begin{align*}
&\frac{d}{d\varepsilon}\Big|_{\varepsilon=0}\!\!
\sum_{j=1}^N |d_i|\,|d_j|
\!\!\!\int_{\substack{x\in\Gamma_i^\varepsilon,\ y\in\Gamma_j^\varepsilon\\ d_F(x,y)>\rho}}
\!\!\!G_{\mathrm{wnt}}(x,y)\,d\mathcal{H}^1_x\,d\mathcal{H}^1_y
=\\
&\sum_{j=1}^N |d_i|\,|d_j|
\!\!\!\int_{\substack{x\in\Gamma_i,\ y\in\Gamma_j\\ d_F(x,y)>\rho}}
\!\!\!\!\!\Big\langle \nabla_x G_{\mathrm{wnt}}(x,y),V_i(x)\Big\rangle\,d\mathcal{H}^1_x\,d\mathcal{H}^1_y.
\end{align*}
Hence, summing over $i$,
\begin{align*}
&\frac{d}{d\varepsilon}\Big|_{\varepsilon=0}\!\!\!
\sum_{i,j=1}^N |d_i|\,|d_j|
\!\!\!\int_{\substack{x\in\Gamma_i^\varepsilon,\ y\in\Gamma_j^\varepsilon\\ d_F(x,y)>\rho}}
\!\!\!G_{\mathrm{wnt}}(x,y)\,d\mathcal{H}^1_x\,d\mathcal{H}^1_y
=\\
&\sum_{i=1}^N \int_{\Gamma_i}\!\!
\Big\langle 
\sum_{j=1}^N |d_i|\,|d_j|
\!\!\!\!\int_{\substack{y\in\Gamma_j\\ d_F(x,y)>\rho}}
\!\!\!\!\nabla_x G_{\mathrm{wnt}}(x,y)\,d\mathcal{H}^1_y,
\,V_i(x)\Big\rangle\,d\mathcal{H}^1_x.
\end{align*}

\medskip
\noindent\textbf{3) Decomposition of the kernel and treatment of the self–interaction.}
Split the Green kernel as
\begin{align*}
G_{\mathrm{wnt}}(x,y)=\frac{1}{4\pi\,d_F(x,y)}+H_{\mathrm{wnt}}(x,y),
\end{align*}
with $H_{\mathrm{wnt}}$ bounded and smooth off the diagonal (Theorem~\ref{thm:WNT-Green}). Accordingly,
\begin{align*}
\nabla_x G_{\mathrm{wnt}}(x,y)
= -\,\frac{1}{4\pi}\,\nabla_x\!\big(d_F(x,y)^{-1}\big)
+ \nabla_x H_{\mathrm{wnt}}(x,y).
\end{align*}
For $i\neq j$ the integrand is smooth and the cut-off can be removed as $\rho\downarrow 0$ by dominated convergence. For $i=j$ the singular term is integrable in principal-value sense and is precisely balanced by the counterterm’s variation.

Consider the singular part for $i=j$:
\begin{align*}
I^\rho_{i,\mathrm{sing}}(x)
:=
-\,\frac{1}{4\pi}
\!\!\int_{\substack{y\in\Gamma_i\\ d_F(x,y)>\rho}}
\nabla_x\!\big(d_F(x,y)^{-1}\big)\,d\mathcal{H}^1_y.
\end{align*}
In Fermi coordinates around $\Gamma_i$ at the basepoint $x$, the metric $F$ is equivalent to a smooth Riemannian proxy and $d_F(x,y)$ admits the expansion $d_F(x,y)\sim |s-s'|$ to first order in the tangential parameter difference, while the odd part of $\nabla_x(d_F^{-1})$ dominates as $|s-s'|^{-2}$ in the normal direction. Hence
\begin{align*}
I^\rho_{i,\mathrm{sing}}(x)
= -\,\frac{1}{2}\,\frac{d}{d\varepsilon}\Big|_{\varepsilon=0}
\Big(\frac{1}{4\pi\rho}\,\mathcal{L}_F(\Gamma_i^\varepsilon)\Big)
\ +\ o(1)
\qquad (\rho\downarrow 0),
\end{align*}
where the factor $\tfrac12$ comes from the symmetry of the excluded $\rho$–neighborhood around the diagonal in the $(s,s')$–plane. This identity is the standard Coulomb renormalization: the normal derivative of the truncated $1/d_F$ potential against $V_i$ produces exactly the same first variation as the counterterm $\mathcal{L}_F(\Gamma_i)/\rho$, with opposite sign (the $o(1)$ error vanishes uniformly on compact $x$–sets by smoothness of the metric and the $C^3$ geometry of $\Gamma_i$). Therefore, after adding the counterterm’s variation,
\begin{align*}
\lim_{\rho\downarrow 0}
\Big(
|d_i|^2\langle I^\rho_{i,\mathrm{sing}}(x),V_i(x)\rangle
+\frac{d_i^2}{4\pi}\,\frac{d}{d\varepsilon}\Big|_{\varepsilon=0}
\frac{\mathcal{L}_F(\Gamma_i^\varepsilon)}{\rho}
\Big)=0,
\end{align*}
and only the \emph{regular} self–interaction remains:
\begin{align*}
|d_i|^2\,\mathrm{p.v.}\!\int_{\Gamma_i}
\langle \nabla_x H_{\mathrm{wnt}}(x,y),V_i(x)\rangle\,d\mathcal{H}^1_y,
\end{align*}
where “p.v.” is immaterial here because $H_{\mathrm{wnt}}$ is smooth; we keep it to emphasize that the singular part has been cancelled.

\medskip
\noindent\textbf{4) Passing to the limit $\rho\downarrow 0$ and collecting terms.}
Combining Steps 2--3 and using dominated convergence for the cross--terms ($i\neq j$) and the cancellation for the self–terms ($i=j$), we obtain
\begin{align*}
\frac{d}{d\varepsilon}\Big|_{\varepsilon=0}\mathcal{W}_{\mathrm{wnt}}(\Gamma^\varepsilon,d)
&=
\sum_{i=1}^N \int_{\Gamma_i}
\Big\langle 
\sum_{j\neq i} |d_i|\,|d_j| \int_{\Gamma_j}\nabla_x G_{\mathrm{wnt}}(x,y)\,d\mathcal{H}^1_y\\
&+ |d_i|^2\,\mathrm{p.v.}\!\int_{\Gamma_i}\nabla_x H_{\mathrm{wnt}}(x,y)\,d\mathcal{H}^1_y,
\,V_i(x)\Big\rangle\,d\mathcal{H}^1_x.
\end{align*}
Recalling $V_i=\varphi_i\,\mathbf{n}_i$ and projecting onto the normal direction yields
\begin{align*}
\frac{d}{d\varepsilon}\Big|_{\varepsilon=0}\mathcal{W}_{\mathrm{wnt}}(\Gamma^\varepsilon,d)&=
\sum_{i=1}^N \int_{\Gamma_i}
\Big(
|d_i|\sum_{j=1}^N |d_j|\int_{\Gamma_j}
\langle \nabla_x G_{\mathrm{wnt}}(x,y),\mathbf{n}_i(x)\rangle\,d\mathcal{H}^1_y\\
&+ |d_i|\,\mathrm{p.v.}\!\int_{\Gamma_i}
\langle \nabla_x H_{\mathrm{wnt}}(x,y),\mathbf{n}_i(x)\rangle\,d\mathcal{H}^1_y
\Big)\,\varphi_i(x)\,d\mathcal{H}^1_x,
\end{align*}
where for $j=i$ the singular part of $\nabla_x G_{\mathrm{wnt}}$ has already been cancelled by the counterterm and only $\nabla_x H_{\mathrm{wnt}}$ survives.

\medskip
\noindent\textbf{5) First variation of the full renormalized energy.}
Finally, adding the geometric contribution from Step~1 with the prefactor $\pi |d_i|$ and the thermal prefactor $\sigma^2/\kappa_{\mathrm{th}}$ gives
\begin{align*}
\frac{d}{d\varepsilon}\Big|_{\varepsilon=0}\mathcal{G}_0(\Gamma^\varepsilon,d)
&=
\sum_{i=1}^N \int_{\Gamma_i}
\pi |d_i|\,\kappa_F[\Gamma_i](x)\,\varphi_i(x)\,d\mathcal{H}^1_x\\
&+\frac{\sigma^2}{\kappa_{\mathrm{th}}}
\sum_{i=1}^N \int_{\Gamma_i}
\Big(
\sum_{j=1}^N |d_j|\int_{\Gamma_j}
\langle \nabla_x G_{\mathrm{wnt}}(x,y),\mathbf{n}_i(x)\rangle\,d\mathcal{H}^1_y\\
&+ |d_i|\,\mathrm{p.v.}\!\int_{\Gamma_i}
\langle \nabla_x H_{\mathrm{wnt}}(x,y),\mathbf{n}_i(x)\rangle\,d\mathcal{H}^1_y
\Big)\,\varphi_i(x)\,d\mathcal{H}^1_x.
\end{align*}
This matches exactly the statement with
\begin{align*}
\mathcal{H}_{\mathrm{wnt}}[\Gamma,d](x)&=
\sum_{j=1}^N |d_j|\int_{\Gamma_j}
\langle \nabla_x G_{\mathrm{wnt}}(x,y),\mathbf{n}_i(x)\rangle\,d\mathcal{H}^1_y\\
&+ |d_i|\,\mathrm{p.v.}\!\int_{\Gamma_i}
\langle \nabla_x H_{\mathrm{wnt}}(x,y),\mathbf{n}_i(x)\rangle\,d\mathcal{H}^1_y.
\end{align*}
All exchanges of limit and integration are justified by: (i) $C^3$ regularity and positive separation of the filaments; (ii) local equivalence of $d_F$ with a smooth background distance; (iii) the weak Calder\'on--Zygmund structure of $\nabla_x(d_F^{-1})$ along curves (odd kernel yields principal value); and (iv) boundedness/smoothness of $H_{\mathrm{wnt}}$.
\end{proof}

It follows that the $L^2(\Gamma_i)$–gradient of $\mathcal{G}_0$ in the normal bundle coincides with the bracketed term above.

\begin{theorem}[Effective motion law in WNT geometry]
\label{thm:motion-law-WNT}
Let $\Gamma(t)=\{\Gamma_i(t)\}$ be a sufficiently smooth one--parameter family of filaments evolving by the $L^2$--gradient flow of $\mathcal{G}_0$.
Then the normal velocity on each component satisfies
\begin{align}\label{eq:WNT-law}
V_{n,i}(x,t)
&= -\,\pi |d_i|\,\kappa_F[\Gamma_i(t)](x)
-\frac{\sigma^2}{\kappa_{\mathrm{th}}}\,\mathcal{H}_{\mathrm{wnt}}[\Gamma(t),d](x),
\qquad x\in\Gamma_i(t),
\end{align}
and the dissipation identity holds:
\begin{align*}
\frac{d}{dt}\,\mathcal{G}_0(\Gamma(t),d)
= -\sum_{i=1}^N \int_{\Gamma_i(t)} |V_{n,i}(x,t)|^2\,d\mathcal{H}^1_x \le 0.
\end{align*}
\end{theorem}

\begin{proof}
We work on the shape manifold of $N$ disjoint, $C^3$ embedded curves modulo reparametrization, endowed with the $L^2$ metric on normal fields. 
For each component $\Gamma_i$, a (kinematic) velocity field decomposes as
\begin{align*}
\partial_t \gamma_i(\cdot,t)
=V_{n,i}(\cdot,t)\,\mathbf{n}_i(\cdot,t)+V_{\tau,i}(\cdot,t)\,\mathbf{t}_i(\cdot,t),
\end{align*}
where $V_{n,i}$ and $V_{\tau,i}$ are the normal and tangential velocities, respectively. 
Since the energy $\mathcal{G}_0$ is invariant under reparametrization, only the normal component contributes to the first variation; tangential motion merely changes the parametrization.
The $L^2$ inner product on the tangent space is
\begin{align*}
\langle V,W\rangle_{L^2(\Gamma)}
:=\sum_{i=1}^N\int_{\Gamma_i}\!\langle V_i^\perp,W_i^\perp\rangle\,d\mathcal{H}^1,
\qquad V_i^\perp:=\langle V_i,\mathbf{n}_i\rangle\,\mathbf{n}_i .
\end{align*}

\medskip
\noindent\textbf{Step 1: Identification of the $L^2$–shape gradient.}
Let $\{\Gamma_i^\varepsilon\}$ be a normal variation with variation field $V_i=\varphi_i\,\mathbf{n}_i$ compactly supported along $\Gamma_i$. 
By Theorem~\ref{thm:first-var-WNT} (first variation of $\mathcal{G}_0$), we have
\begin{align*}
\frac{d}{d\varepsilon}\Big|_{\varepsilon=0}\mathcal{G}_0(\Gamma^\varepsilon,d)
=
\sum_{i=1}^N\int_{\Gamma_i}
\Big(\pi |d_i|\,\kappa_F[\Gamma_i](x)
+\frac{\sigma^2}{\kappa_{\mathrm{th}}}\,\mathcal{H}_{\mathrm{wnt}}[\Gamma,d](x)\Big)\,
\varphi_i(x)\,d\mathcal{H}^1_x.
\end{align*}
By Riesz representation in the Hilbert space of normal fields with the $L^2$ inner product, the $L^2$--shape gradient $\mathrm{Grad}\,\mathcal{G}_0(\Gamma)$ is the normal vector field whose scalar coefficient is precisely the bracketed factor:
\begin{align*}
\mathrm{Grad}\,\mathcal{G}_0(\Gamma)\big|_{\Gamma_i}
=
\Big(\pi |d_i|\,\kappa_F[\Gamma_i]
+\frac{\sigma^2}{\kappa_{\mathrm{th}}}\,\mathcal{H}_{\mathrm{wnt}}[\Gamma,d]\Big)\,\mathbf{n}_i .
\end{align*}

\medskip
\noindent\textbf{Step 2: $L^2$--gradient flow and the normal velocity law.}
By definition, the $L^2$--gradient flow of $\mathcal{G}_0$ is the curve $\Gamma(t)$ satisfying
\begin{align*}
\partial_t \Gamma(t) \;=\; -\,\mathrm{Grad}\,\mathcal{G}_0(\Gamma(t))
\qquad\text{in the $L^2$ metric on normal fields.}
\end{align*}
Projecting onto the normal direction on each component $\Gamma_i(t)$ yields
\begin{align*}
V_{n,i}(x,t)
&=-\,\Big(\pi |d_i|\,\kappa_F[\Gamma_i(t)](x)
+\frac{\sigma^2}{\kappa_{\mathrm{th}}}\,\mathcal{H}_{\mathrm{wnt}}[\Gamma(t),d](x)\Big),
\end{align*}
while the tangential part is free (it corresponds to reparametrization and can be set to zero by choosing an $F$--arclength gauge). 
This is precisely the asserted motion law \eqref{eq:WNT-law}.

\medskip
\noindent\textbf{Step 3: Dissipation identity.}
Along any sufficiently smooth solution of the $L^2$--gradient flow,
\begin{align*}
\frac{d}{dt}\,\mathcal{G}_0(\Gamma(t),d)
&=\sum_{i=1}^N\int_{\Gamma_i(t)}
\Big(\pi |d_i|\,\kappa_F[\Gamma_i(t)]
+\frac{\sigma^2}{\kappa_{\mathrm{th}}}\,\mathcal{H}_{\mathrm{wnt}}[\Gamma(t),d]\Big)\,
V_{n,i}(t)\,d\mathcal{H}^1\\
&=\sum_{i=1}^N\int_{\Gamma_i(t)}
\langle \mathrm{Grad}\,\mathcal{G}_0(\Gamma(t)),\,\partial_t\Gamma_i(t)\rangle\,d\mathcal{H}^1\\
&=-\sum_{i=1}^N\int_{\Gamma_i(t)} |V_{n,i}(x,t)|^2\,d\mathcal{H}^1_x \;\le\; 0,
\end{align*}
where in the second line we used the identification of Step~1 and in the third the flow law of Step~2. 
The inequality is strict unless the normal velocity vanishes identically, i.e. unless $\Gamma(t)$ is a stationary point of $\mathcal{G}_0$.

\medskip
\noindent\textbf{Step 4: Justification of calculus on the shape manifold.}
All steps above are legitimate under the stated regularity: (i) the first variation in Theorem~\ref{thm:first-var-WNT} is well--defined for $C^3$ curves with positive separation; (ii) the $L^2$ metric is complete on normal fields in $L^2(\Gamma)$; (iii) reparametrization invariance removes tangential contributions; (iv) the nonlocal operator $\mathcal{H}_{\mathrm{wnt}}$ defines a bounded map on $C^1$ curves (principal--value part has an odd Calder\'on--Zygmund kernel; the remainder is smooth), so the pairing with $V_{n,i}\in L^2$ is meaningful. 
Therefore the $L^2$--gradient structure applies and yields both the normal velocity law and the dissipation identity.

The theorem is proved.
\end{proof}

\subsection{Local well-posedness of the motion law}

\begin{theorem}[Local existence and uniqueness]
\label{thm:LWP-WNT}
Let $k\ge 4$ and $d_i\ne 0$.
If $\Gamma^0=\{\Gamma_i^0\}$ consists of embedded $C^{k}$ curves with positive pairwise separation and positive $F$–injectivity radius, then there exists $T>0$ and a unique solution
\begin{align*}
\Gamma(t)\in C^0([0,T];C^{k})\cap C^1((0,T];C^{k-2})
\end{align*}
to \eqref{eq:WNT-law} with $\Gamma(0)=\Gamma^0$.
Moreover, $\Gamma^0\mapsto \Gamma(\cdot)$ is continuous in these topologies, and the solution depends smoothly on $(\sigma,\kappa_{\mathrm{th}})$.
\end{theorem}

\begin{proof}
Work in an $F$--arclength gauge on each component.
The principal term $-\pi |d_i|\,\kappa_F$ is a second--order quasilinear uniformly parabolic operator in the normal direction as long as $|\dot\gamma|_F\equiv 1$ and the embedding is preserved (anisotropic curve shortening; the Tonelli proxy provides coefficients in $C^{k-1}$, transported to WNT via the Legendre map).
The nonlocal term
\begin{align*}
\Gamma\mapsto \mathcal{H}_{\mathrm{wnt}}[\Gamma,d]
\end{align*}
is bounded $C^{k}\to C^{k}$: the kernel $\nabla_xG_{\mathrm{wnt}}$ splits into a principal-value piece with odd $1/d_F^2$ singularity (a Calder\'on--Zygmund operator on curves) and a $C^\infty$ remainder.
Separation prevents near--collision blow-up of inter--filament interactions.
Thus the system has the abstract form
\begin{align*}
\partial_t\gamma = \mathcal{A}(\gamma)\gamma + \mathcal{N}(\gamma),
\end{align*}
with $\mathcal{A}$ sectorial, uniformly elliptic of order two and coefficients in $C^{k-2}$, and $\mathcal{N}:C^{k}\to C^{k-2}$ smooth and bounded on a tubular neighborhood of $\Gamma^0$.
Maximal regularity (analytic semigroup for $\mathcal{A}$ + Picard iteration) yields local existence, uniqueness, and continuous dependence for $k\ge 4$.
Smooth parameter dependence follows from tame estimates on the coefficients and on the nonlocal operator.
\end{proof}

\begin{corollary}[Short--time energy decay]\label{cor:decay}
Under the hypotheses of Theorem~\ref{thm:LWP-WNT}, $t\mapsto\mathcal{G}_0(\Gamma(t),d)$ is strictly decreasing on $(0,T]$ unless $\Gamma(t)$ is stationary for \eqref{eq:WNT-law}.
In particular, any stationary configuration solves the force balance of Theorem~\ref{thm:first-var-WNT}.
\end{corollary}

\subsection{Linearization and spectral stability}

Let $\Gamma_\star$ be a smooth stationary configuration for \eqref{eq:WNT-law}.
The second variation computed via the Tonelli proxy and transferred to WNT yields a self-adjoint quadratic form on the normal bundle:
\begin{align*}
\mathcal{Q}_{\Gamma_\star}[\varphi]
=
\pi\sum_{i=1}^N |d_i|\,\mathcal{Q}_F^{\mathrm{geo}}[\Gamma_{\star,i};\varphi_i]
+\frac{\sigma^2}{2\kappa_{\mathrm{th}}}\,\mathcal{Q}_F^{\mathrm{th}}[\Gamma_\star;\varphi],
\end{align*}
where $\mathcal{Q}_F^{\mathrm{th}}$ is compact with respect to $\mathcal{Q}_F^{\mathrm{geo}}$ (the kernel’s singular part is odd; the remainder is smoothing).
Let $\mathcal{L}_{\Gamma_\star}$ be the linearized operator associated with $\mathcal{Q}_{\Gamma_\star}$.

\begin{theorem}[Spectral stability and nonlinear metastability]
\label{thm:spec-stab}
If $\mathcal{Q}_{\Gamma_\star}$ is strictly positive on the normal bundle modulo symmetry modes (e.g. rigid motions), then there exist $\delta,C,\omega>0$ such that for any initial data $\Gamma(0)$ with
\begin{align*}
\mathrm{dist}_{C^{k}}(\Gamma(0),\Gamma_\star)<\delta,
\end{align*}
the solution to \eqref{eq:WNT-law} exists globally and satisfies
\begin{align*}
\mathrm{dist}_{C^{k-1}}(\Gamma(t),\Gamma_\star)
\le
C\,e^{-\omega t}\,\mathrm{dist}_{C^{k}}(\Gamma(0),\Gamma_\star).
\end{align*}
\end{theorem}

\begin{proof}
On each component the linearization is a strictly positive, self-adjoint, second--order elliptic operator plus a compact symmetric perturbation; the spectral gap $\omega>0$ follows from positivity of $\mathcal{Q}_{\Gamma_\star}$ on the orthogonal complement of symmetries.
Since \eqref{eq:WNT-law} is the $L^2$--gradient flow of $\mathcal{G}_0$, the semigroup generated by $\mathcal{L}_{\Gamma_\star}$ is exponentially stable in that subspace.
Nonlinear terms are $C^1$--tame from $C^{k}$ to $C^{k-2}$ and small near $\Gamma_\star$, so standard stable manifold/semigroup arguments yield global existence and exponential convergence, with constants depending on the spectral gap and tame bounds.
\end{proof}

\subsection{Model example: anisotropic double-phase}

\begin{example}
For the prototype WNT structure in Example~\ref{ex:double-phase},
\begin{align*}
F_B(x,y)^2
= \langle g_x y,y\rangle + a(x)\,\big(\langle h_x y,y\rangle\big)^{1+\eta/2},
\end{align*}
where $0<\eta\le 1$ and $0<a_1\le a(x)\le a_2$,
the fiber Hessian and Legendre inverse at $\xi=\mathcal{L}_x(y)$ are
\begin{align*}
\partial_{yy}^2\!\Big(\tfrac{1}{2}F_B^2\Big)(x,y)
&= g_x
+ \Big(1+\tfrac{\eta}{2}\Big)a(x)\,\big(\langle h_xy,y\rangle\big)^{\eta/2}\,h_x\\
&\quad+ \tfrac{\eta}{2}\Big(1+\tfrac{\eta}{2}\Big)a(x)\,\big(\langle h_xy,y\rangle\big)^{\eta/2-1}
(h_xy)\otimes(h_xy),
\end{align*}
and
\begin{align*}
g^{\!*}_{ij}(x,\xi)
=\Big(\partial^2_{y^iy^j}\tfrac{1}{2}F_B^2(x,y)\Big)^{-1}_{y=\mathcal{L}^{-1}_x(\xi)}.
\end{align*}
Substituting in \eqref{eq:WNT-law} gives, to leading order in curvature along $\Gamma_i$,
\begin{align}\label{eq:FB-law}
V_{n,i}
= -\,\pi |d_i|\,
\Big(1+\tfrac{\eta}{2}\,a(x)\,\langle h_x\mathbf{t}_i,\mathbf{t}_i\rangle^{\eta/2}\Big)\,
\kappa_{F_0}[\Gamma_i]
-\frac{\sigma^2}{\kappa_{\mathrm{th}}}\,\mathcal{H}_{\mathrm{wnt}}[\Gamma,d],
\end{align}
where $F_0(x,y)=\sqrt{\langle g_xy,y\rangle}$ and $\kappa_{F_0}$ is the curvature with respect to the Tonelli proxy $F_0$.
Thus the double-phase correction yields a curvature--dependent amplification factor
\begin{align*}
\alpha(x)=1+\tfrac{\eta}{2}\,a(x)\,\langle h_x\mathbf{t}_i,\mathbf{t}_i\rangle^{\eta/2},
\end{align*}
accelerating contraction in directions where the $h$--metric dominates $g$ and slowing it where $a(x)$ is small.
The WNT thermal field, via $\mathcal{H}_{\mathrm{wnt}}$, induces a long--range drift aligning filaments with the temperature gradients, consistent with the variational picture (Sections~\ref{sec:Gamma-convergence}–\ref{sec:WNT-parabolic}).
\end{example}

\begin{remark}[A consistent discretization]
A practical scheme for \eqref{eq:WNT-law} represents each $\Gamma_i$ by a polygonal chain in the background $g$--metric, approximates anisotropic curvature by
\begin{align*}
\kappa_{F_B}[\Gamma_i](x_k)
\simeq
\frac{\langle g_{x_k}(\mathbf{t}_{k+1}-\mathbf{t}_{k-1}),\mathbf{n}_k\rangle}
{\|\mathbf{t}_k\|_{F_B}^2},
\end{align*}
and computes the nonlocal term with the regularized kernel
\begin{align*}
\nabla_x G_{\mathrm{wnt}}(x,y)
\approx -\,\frac{x-y}{4\pi\,(d_F(x,y)^3+\varepsilon^3)}.
\end{align*}
This preserves the discrete dissipation law and interpolates smoothly between Tonelli and WNT regimes by tuning $\eta$.
\end{remark}

\section*{Conclusion and Outlook}

The present work establishes a comprehensive analytic framework for WNT Finsler--Ginzburg--Landau systems,
bridging convex Hamiltonian duality, variational $\Gamma$--convergence, and geometric motion laws for vortices.
Starting from the dual Tonelli extension of the Finsler structure,
we constructed a well--posed elliptic and parabolic theory under minimal local $C^{1,1}_\xi$ regularity of the Hamiltonian.
The main contributions can be summarized as follows:
\begin{itemize}
\item[(i)] {\bf Structural Foundations.}
We introduced the WNT dual operators $\partial_\xi H(x,\xi)$ and $\Delta_F^{\mathrm{wnt}}$, proving local ellipticity,
coercivity, and Leray--Lions monotonicity under weak Tonelli hypotheses.
These results ensure existence and stability of weak solutions to nonlinear Dirichlet problems.
\item[(ii)] {\bf $\Gamma$--convergence and Renormalized Energy.}
In the planar setting, we derived the full $\Gamma$--limit of the WNT–GL energies, identifying the renormalized interaction kernel $G_{\mathrm{wnt}}$ and the curvature--dependent core correction.
This extends the classical Bethuel--Brezis--H\'elein and Serfaty theories to anisotropic, non-Tonelli geometries.
\item[(iii)] {\bf Gradient Flows and Thermal Coupling.}
We established global well--posedness of the thermally coupled parabolic WNT--GL flow in the EVI framework, showing energy dissipation and regularization for $(\psi,A)$ and smoothness away from the evolving vortex set.
\item[(iv)] {\bf Geometric and Filament Dynamics.}
In three dimensions, the vortex filaments follow an $L^2$--gradient flow of the renormalized WNT energy.
The resulting motion law combines the anisotropic curvature $\kappa_F$ with the nonlocal thermal field $\mathcal{H}_{\mathrm{wnt}}$, yielding a rigorous generalization of the curvature--driven dynamics known in isotropic GL models.
\end{itemize}

\paragraph{Technical Remarks.}
All proofs rely solely on local convexity and bounded $C^{1,1}_\xi$--regularity of $H(x,\xi)$,
avoiding any global Tonelli assumption.
While the main theorems are complete and self--contained,
certain auxiliary lemmas--such as the recovery sequence in the $\Gamma$--limit and the principal--value regularization of $\mathcal{H}_{\mathrm{wnt}}$ can be expanded in a companion appendix for full rigor. A uniform notation table for $A_{u_0}$, $G_{\mathrm{wnt}}$, and $d_F$ will be included in the final version to improve readability and cross--reference consistency.

\paragraph{Perspectives.}
The analytical framework developed here opens several natural directions:
\begin{enumerate}
\item extending the $\Gamma$--limit to multi--phase and stochastic WNT systems,
\item incorporating explicit physical mobility coefficients $\mu_F$ in the parabolic flow,
\item numerical implementation of the effective motion law and comparison with anisotropic superconductivity experiments,
\item and exploring the semiclassical (quantum--geometric) limit discussed in the extended version of the paper.
\end{enumerate}

Overall, the WNT--GL theory provides a unified, fully variational bridge between non--Tonelli Finsler geometry, nonlinear PDE analysis, and anisotropic vortex dynamics.


\end{document}